\documentclass{LMCS}

\def\dOi{10(4:4)2014}
\lmcsheading%
{\dOi}
{1--39}
{}
{}
{Mar.~29, 2014}
{Dec.~\phantom03, 2014}
{}

\ACMCCS{[{\bf Theory of computation}]: Models of computation;
  Semantics and reasoning}

\pdfoutput=1
\makeatletter
\let\fact\@undefined
\let\endfact\@undefined
\makeatother
\usepackage{macros}

\newcommand{\IEEEhspace}[1]{}
\usepackage{amsfonts,amstext,amsmath,amssymb,stmaryrd,mathrsfs,calc,amsthm,mathtools,bold-extra}
\usepackage{bm}%
\newcommand{\nocolon}[1]{}
\usepackage{thmtools}
\theoremstyle{plain}
\newtheorem{theorem}{Theorem}[section]
\newtheorem{lemma}[theorem]{Lemma}
\newtheorem{proposition}[theorem]{Proposition}
\newtheorem{corollary}[theorem]{Corollary}
\newtheorem{fact}[theorem]{Fact}
\theoremstyle{definition}
\newtheorem{definition}[theorem]{Definition}
\newtheorem{example}[theorem]{Example}
\theoremstyle{remark}
\newtheorem{remark}[theorem]{Remark}

\usepackage{subfig,tikz}
\usetikzlibrary{external}
\usetikzlibrary{positioning}
\usetikzlibrary{arrows}
\usetikzlibrary{automata}
\usetikzlibrary{trees}
\usetikzlibrary{calc}
\tikzstyle{stt}=[shape=circle,fill=black!12,draw=black!80]
\makeatletter
\def\@listi{\leftmargin\leftmargini
               \topsep 3\p@ \@plus\p@ \@minus\p@
               \parsep 2\p@ \@plus\p@ \@minus\p@
               \itemsep \parsep}
\makeatother
\providecommand{\citep}{\cite}
\providecommand{\citet}{\cite}

\usepackage{url,hyperref}\providecommand{\autoref}{\ref}
\usepackage[backgroundcolor=green!20, textsize=scriptsize]{todonotes}

\begin{document}
\renewcommand{\figureautorefname}{Figure}
\renewcommand{\sectionautorefname}{Section}
\renewcommand{\subsectionautorefname}{Section}
\renewcommand{\subsubsectionautorefname}{Section}

\title[The Power of Priority Channel Systems]
      {The Power of Priority Channel Systems\rsuper*}

\author[C.~Haase]{Christoph Haase}
\author[S.~Schmitz]{Sylvain Schmitz}
\author[Ph.~Schnoebelen]{Philippe Schnoebelen}
\address{LSV, ENS Cachan \& CNRS \& INRIA, France}
\email{\{haase,schmitz,phs\}@lsv.ens-cachan.fr}
\thanks{Work partially funded by the \textsc{ReacHard} project ANR 11 BS02 001 01.}

\titlecomment{{\lsuper*}An extended abstract of this work first appeared in the
  proceedings of the 24th International Conference on Concurrency
  Theory (Concur'13)~\cite{HSS13}.%
}
\begin{abstract}
  We introduce Priority Channel Systems, a new natural class of
  channel systems where messages carry a numeric priority and where
  higher-priority messages can supersede lower-priority messages
  preceding them in the fifo communication buffers.  The decidability
  of safety and inevitability properties is shown via the introduction
  of a \emph{priority embedding}, a well-quasi-ordering that has not
  previously been used in well-structured systems.  We then show how
  Priority Channel Systems can compute fast-growing functions and
  prove that the aforementioned verification problems are {$\mathbf
    F_{\varepsilon_0}$}-complete.
\end{abstract}

\maketitle
\section{Introduction}

\emph{Channel systems} are a family of distributed models where
concurrent agents communicate via usually unbounded fifo communication
buffers called ``channels.''  An agent of a channel system is modeled
by a finite-state controller, and when taking a transition an agent
can read messages from the channel or write into it. These models have
turned out to be well-suited for the formal specification and
algorithmic analysis of communication protocols and concurrent
programs~\cite{pachl87,boigelot99b,bouajjani99b,cece2005b,muscholl2010}. 
They are also a fundamental model of computation, closely related to
Post's tag systems. In all generality, channel systems are a Turing
powerful model, which implies that most of their decision problems are
undecidable.

A particularly interesting decidable and widely studied class of
channel systems are the so-called
\emph{lossy channel systems} (LCSs), where  channels are unreliable and may lose
messages, see \emph{e.g.}~\cite{cece95,abdulla96b,BMOSW-fac2012}. For
LCSs, several important behavioral properties such as safety or
inevitability are decidable. This is because, due to the lossy
behavior of their channels, these systems are \emph{well-structured}:
transitions are monotonic with respect to a decidable
well-quasi-ordering of the configuration
space~\cite{wsts2000,FinSch-WSTS,concur}.  Beyond their applications in
verification, LCSs have turned out to be an important
automata-theoretic tool for decidability or hardness in areas like
Timed Automata, Metric Temporal Logic, modal
logics, \emph{e.g.}~\cite{abdulla-icalp05,kurucz06,ouaknine2007,lasota2008}.
Moreover, they are also a
fundamental model of computation capturing the
$\bfF_{\omom}$-complexity level in the fast-growing complexity
hierarchy~\cite{arXiv/Schmitz13},
see~\cite{CS-lics08,SS-icalp11}.

Lossy channel systems do not provide an adequate way to model systems
or protocols that treat messages discriminatingly according to some
specified rule set. An example is the prioritization of messages,
which is central to ensuring \emph{quality of service} (QoS)
properties in networking architectures, and is usually implemented by
allowing for tagging messages with some relative priority. For
instance, the Differentiated Services (DiffServ) architecture
described in RFC 2475~\cite{RFC2475}, which enables QoS on modern IP
networks, allows for a field specifying the relative priority of an IP
packet with respect to a finite set of priorities, and network links
may decide to arbitrarily drop IP packets of lower priority in favor
of higher priority packets once the network congestion reaches a
critical point. Another example of a similar priority-based policy
arises in the context of ATM networks, where priorities are expressed
via a single Cell Loss Priority bit in order to allow for giving
preference (by dropping low-priority packages) to audio or video over
less time-critical data~\cite{leboudec92}.

Inspired by the aforementioned types of protocols, in this paper we
introduce \emph{priority channel systems} (PCSs), a family of channel
systems where each message is equipped with a priority level, and
where higher-priority messages can supersede lower-priority messages
by dropping them. Priority channel systems rely on the
\emph{prioritized superseding ordering}, a novel ordering that
generalizes Higman's subword ordering and has not been considered
before in the area of well-structured systems. It is however closely
related to the gap-embedding considered in~\cite{schutte85}. Showing
it to be a well-quasi-ordering entails, among others, showing the
decidability of safety and termination for PCSs. We complement our
decidability results by showing that these problems
become undecidable for channel systems that build upon more
restrictive priority mechanisms, supporting the design choices made
for our model.

\subsection{Structure of this Paper}

This paper can roughly be divided into two parts. In the first part, we define
priority channel systems, explore this new model and analyze its power in
complexity-theoretical terms.
Beginning in \autoref{sec-related}, the second part relates
priority channel systems in the broadest sense to related models or
mathematical objects found in the literature.

In more detail, in~\autoref{sec-pcs} we provide an \emph{at-a-glance}
introduction to a simplified model of priority channel systems. This
allows us to discuss on a high level the ideas behind our model, the
main theorems, and the main algorithmic problems that we consider in
this paper. We outline the decidability of fundamental decision
problems via the framework of well-structured
systems. \autoref{sec-embedding} is then devoted to proving
well-quasi-ordering properties of the prioritized superseding ordering
which underlies priority channel systems. To this end, we characterize
the superseding ordering via \emph{priority embeddings}, which is an
analogue and can in fact be seen as a generalization of Higman's
subword embedding. Using techniques from~\cite{SS-icalp11,schutte85},
we show in \autoref{sec-fg} an $\bfF_{\epsz}$ upper bound on the
complexity of PCS verification, far higher than the
$\bfF_{\omom}$-complete complexity known for LCSs. We then prove in
\autoref{sec-hardy} a matching lower bound and this is the main
technical result for PCSs of this paper: building upon techniques
developed for less powerful
models~\cite{CS-lics08,phs-mfcs2010,HSS-lics2012}, we show how PCSs
can robustly simulate the computation of the fast growing functions
$F_\alpha$ and their inverses for all ordinals $\alpha$ up to
$\epsz$. This gives a precise measure of the expressive power of PCSs.

In the second part of the paper, we first show in
\autoref{sec-related} that other natural choices of models of channel
systems with priority mechanisms different from ours lead to
undecidability of verification problems.  We then turn to lossy channel
systems and show in \autoref{ssec-holcs} that, although PCSs are not
an extension of LCSs, they can very easily simulate LCS computations.
In fact, we show how \emph{higher-order} LCS models, which generalize the
dynamic LCS from~\cite{AAC-fsttcs2012}, also embed into PCSs; this
uses a more involved encoding of higher-order configurations and
rules.  Applications of the priority embedding to other
well-quasi-ordered data structures such as depth-bounded trees found
in the literature are subsequently discussed in \autoref{ssec-trees}.

\section{Priority Channel Systems}\label{sec-pcs}
In this section, we formally introduce Priority Channel Systems and
give an overview about the decision problems we consider in this
paper.
\begin{definition}
  \label{def-pcs}
  For every $d\in\Nat$, the \emph{level-$d$ priority alphabet} is
  $\Sigma_d\eqdef\{0,1,\ldots,d\}$. A \emph{level-$d$ priority channel
    system} ($d$-PCS) is a tuple $S=(\Sigma_d,\Ch,Q,\Delta)$,
  where $\Sigma_d$ is as above, $\Ch=\{\ttc_1,\ldots,\ttc_m\}$ is a
  set of $m$ \emph{channel names}, $Q=\{q_1,q_2,\ldots\}$ is a finite
  set of \emph{control states}, and $\Delta\subseteq
  Q\times\Ch\times\{!,?\}\times\Sigma_d\times Q$ is a set of
  \emph{transition rules}.
\end{definition}
\begin{figure}[hbt]
\centering
\scalebox{1.0}{
  \begin{tikzpicture}[->,>=stealth',shorten >=1pt,node distance=4cm,thick,auto,bend angle=30]
\path
(0,0) node [shape=circle,fill=black!12,draw=black!80] (q1) {$p$}
(2,0) node [shape=circle,fill=black!12,draw=black!80] (q2) {$q$}
;
\path (q1) edge[bend left=10] node {$\ttc \, ! \, 1$} (q2) ;
\path (q2) edge[bend left=10] node {$\ttc \, ? \, 3$} (q1) ;
\path (q1) edge[loop left] node {$\ttc \, ! \, 0$} (q1);
\path (q2) edge[loop right] node {$\ttc \, ! \, 3$} (q2);
\node[right=7em of q2,fill=black!12,text width=5em,text height=0.8em,align=flush right] (c){$0\,3\,0\,0\;$};
\draw[-,thick,color=black!80]
    (c.north west) -- (c.north east)
    (c.south west) -- (c.south east);
\node[left=0.3em of c]{$\ttc\,$:};
  \end{tikzpicture}
}%
\caption{A simple single-channel $3$-PCS.}
\label{fig-ex-pcs}
\end{figure}
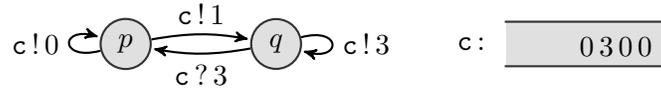
For the sake of a simplified introduction to PCSs in this section, the
alphabet of a PCS abstracts away from actual message contents and only
consists of natural numbers that indicate the priority of a message, where
$d$ is a message of highest and $0$ of lowest priority,
respectively. A treatment of more general alphabets is deferred to
\autoref{sec-embedding}. The simple alphabet introduced here is
however sufficient in order to show the lower bounds in
\autoref{sec-hardy}. Moreover, from our definition it follows that
Priority Channel Systems consist of a single process, which is
sufficient for our purposes in this paper, since systems made of
several concurrent components can be represented by a single process
obtained as an asynchronous product of the components.

\autoref{fig-ex-pcs} depicts a $3$-PCS with a single channel and
control states $p$ and $q$. (A possible configuration of the channel is depicted alongside.) Informally speaking, when in control state
$q$ the PCS in \autoref{fig-ex-pcs} can non-deterministically loop
while writing the alphabet symbol ``3'' to the channel (to its right end), or switch to
control state $q$ if ``3'' can be read from the channel (from its left end). The key
feature of PCSs is that messages with higher priority can erase
messages with lower priority, \emph{cf.} the formal semantics given
next.

\subsection{Semantics}
\label{ssec-pcs-semantics}
The operational semantics of a PCS $S=(\Sigma_d,\Ch,Q,\Delta)$ is given in terms of a
transition system. We let $\Conf_S\eqdef Q\times(\Sigma_d^*)^m$ be the
set of all \emph{configurations} of $S$, denoted $C,D,\ldots$ in the
following. A configuration $C=(q,x_1,\ldots,x_m)$ records an
instantaneous control state $q\in Q$ and the contents of the $m$
channels, \emph{i.e.}, sequences of messages from $\Sigma_d$. A
sequence $x\in\Sigma_d^*$ has the form $x=a_1 \cdots a_\ell$ and we
let $\len{x}=\ell$. Concatenation is denoted multiplicatively, with
$\epsilon$ denoting the empty sequence.

The labeled transition relation between configurations, denoted
$C\step{\delta}C'$, is generated by the rules in
$\Delta=\{\delta_1,\ldots, \delta_k\}$. From a technical perspective,
it is convenient to define two such transition relations,
denoted $\rstep{}$ and $\hstep{}$.

\subsubsection{Reliable Semantics.}
We start with $\rstep{}$ that corresponds to ``reliable'' steps, or
more correctly steps with no superseding of lower-priority messages.  
As is standard, for a \emph{reading rule} of the form
$\delta=(q,\ttc_i{?}a,q')\in\Delta$, there is a step
$C\rstep{\delta}C'$ if $C=(q,x_1,\ldots,x_m)$ and
$C'=(q',y_1,\ldots,y_m)$ for some $x_1,y_1,\ldots,x_m,y_m$ such that
$x_i=a\,y_i$ and $x_j=y_j$ for all $j\not=i$, while for a \emph{writing
  rule} $\delta=(q,\ttc_i{!}a,q')\in\Delta$, there is a step
$C\rstep{\delta}C'$ if $y_i=x_i\,a$ and $x_j=y_j$ for all $j\not=i$.
These reliable steps correspond to the behavior of queue automata, or
(reliable) channel systems, a Turing-powerful computation
model~\citep{brand83}.

\subsubsection{Internal-Superseding.}
The actual behavior of PCSs is obtained by extending reliable steps
with \emph{internal superseding steps}, denoted $C\hstep{\ttc_i\hh
  k}C'$, which can be performed at any time in an uncontrolled
manner.  Formally, for two words $x,y\in\Sigma_d^*$ and $k>0$ in
$\Nat$, we write $x\hstep{\hh k}y$ $\equivdef$
\begin{enumerate}
\item $x$ can be decomposed as $a_1\cdots a_\ell$ with $\ell>k$;
\item $a_k\leq a_{k+1}$; and
\item $y=a_1\cdots a_{k-1}\,a_{k+1}\cdots a_\ell$.
\end{enumerate}
In other words, the $k$th message in $x$ is superseded by its
immediate successor $a_{k+1}$, with the condition that $a_k$ is not of
higher priority.  We write $x\hstep{}y$ when $x\hstep{\hh k}y$ for
some $k$, and use $x\revhstep{}y$ when $y\hstep{}x$. The transitive
reflexive closure $\revhstep{\ast}$ is called the \emph{superseding
  ordering} and is denoted by $\hleq$. Put differently, $\hstep{}$ is
a rewrite relation over $\Sigma^*_d$ defined by the
following string rewriting
system (see~\citep{book82}):
\begin{equation}
\label{eq-rewr-4-hstep}
  \{a\,a'\to a'\mid 0\leq a\leq a'\leq d\}\;.
\end{equation}

This is extended to steps between configurations by
$C \hstep{\ttc_i\hh k} C'$
$\equivdef$ $C=(q,x_1,\ldots,x_m)$, $C'=(q',y_1,\ldots,y_m)$, $q=q'$, $x_i\hstep{\hh k}y_i$, and $x_j=y_j$ for
$j\not=i$. Furthermore, every reliable step is a valid step: for any
rule $\delta$, $C\hstep{\delta} C'$ if $C\rstep{\delta}C'$, giving
rise to a second transition system associated with $S$:
$\TS_\hh\eqdef(\Conf_S,\hstep{})$. 

\begin{example}\label{ex-internal}
 The following is a valid path in the transition system induced by the
 PCS from \autoref{fig-ex-pcs}:
 \[
 (p,0\,3\,0\,0)
 \hstep{!1} (q,0\,3\,\underline{0}\,0\,1)
 \hstep{\hh 3} (q,\underline{0}\,3\,0\,1)
 \hstep{\hh 1} (q,3\,\underline{0}\,1)
 \hstep{\hh 2} (q,3\,1)\;.
 \]
 Here, underlining is used to show which symbol, if any, is
 superseded in the next step.
\end{example}

\begin{remark}[Lossy Channel Systems]
  Priority channel systems and lossy channel systems are unreliable in
  the sense that uncontrolled rewrites may occur inside the channels.
  In the case of LCSs, arbitrary message losses can take place.  The
  semantics of the two classes of systems thus differ quite a bit,
  however PCSs are easily seen to be able to encode LCS computations
  (by interspersing channel contents with a higher-priority symbol;
  see \autoref{ssec-holcs} and in particular
  \autoref{fig-simul1-holcs}).
\end{remark}

\subsubsection{Write-Superseding Semantics}\label{ssec-write-superseding}
The internal-superseding semantics allows superseding to occur at any
time and anywhere in the channel. Another possible scenario 
considers 
communications going through relays,
network switches, or buffers, which handle incoming
traffic with a so-called \emph{write-superseding} policy, where writes
immediately supersede (\textit{i.e.}, erase) the congested messages in front of them. We develop this aspect here and prove the two semantics to be
essentially equivalent.

Let $S=(\Sigma_d,\Ch,Q,\Delta)$ be a $d$-PCS. We define a new
transition relation, denoted $\pstep{}$, between the configurations of
$S$, giving rise to a transition system $\TS_\ww \eqdef
(\Conf_S,\pstep{})$. The relation $\pstep{}$ is a variant of
$\hstep{}$ obtained by modifying the semantics of writing rules.
Formally, for $\delta=(q,\ttc_i{!}a,q')\in\Delta$, and for two
configurations $C=(q,x_1,\dots,x_m)$ and $C'=(q,y_1,\dots,y_m)$, there
is a step $C\pstep{\delta}C'$ if $x_j=y_j$ for all $j\neq i$ and
$y_i=z\,a$ for a factorization $x_i=z\,z'$ of $x_i$ where
$z'\in\Sigma_a^*$, \textit{i.e.}, where $z'$ only contains messages
from the level-$a$ priority subalphabet. In other words, after
$\ttc_i!a$, the channel will contain a sequence $y_i$ obtained from
$x_i$ by appending $a$ in a way that may drop (erase) any number of
suffix messages with priority $\leq a$, hence the
``$z'\in\Sigma_a^*$'' requirement. The semantics of reading rules is
unchanged so that $C\rstep{\delta}C'$ implies $C\pstep{\delta}C'$.

\begin{example}\label{ex-write-superseding}%
The PCS from \autoref{fig-ex-pcs}
has the following write-superseding run:
\[
           (p,0\,3\,\underline{0\,0})
\pstep{!1} (q,0\,\underline{3\,1})
\pstep{!3} (q,0\,3)
\pstep{!3} (q,\underline{0\,3\,3})
\pstep{!3} (q,3)
\pstep{?3} (p,\epsilon)
\]
where in every configuration we underline the messages that will be
superseded in the next step (and where, for simplicity, we do not
write the full rule $\delta$ on the steps). Observe that
$(p,0\,3\,0\,0)\not\pstep{*}(q,3\,1)$, to be contrasted with the
internal-superseding run $(p,0\,3\,0\,0)\hstep{*}(q,3\,1)$ in
\autoref{ex-internal}. Under write-superseding, the occurrence of $3$
that is initially in the channel is not allowed to erase the $0$ in
front of it.
\qed
\end{example}

Compared to our standard PCS semantics, the write-superseding
semantics adopts a localized viewpoint where the protocol managing
priority levels and handling congestions resides at the sender's end,
and is not distributed all along the channels.

In the rest of this subsection, we show that the write-superseding is
essentially equivalent to the standard semantics, see
\autoref{prop-psteps-hsteps}.  A consequence is that one can freely
choose to adopt either $\TS_\hh$ or $\TS_\ww$ as their favorite
operational semantics for priority channel systems.  In practice, we
find it simpler to design and prove the correctness of some PCS---as
we will in Sections~\ref{sec-hardy} and~\ref{ssec-holcs}---when assuming the write-superseding
semantics since it is less liberal and easier to control. And we find
it simpler to develop the formal theory of PCSs when assuming the
internal-superseding semantics since it is finer-grained.

\begin{proposition}
\label{prop-psteps-hsteps}
Let $C_0=(q,\epsilon,\ldots,\epsilon)$ be a configuration with empty
channels, and $C_f$ be any configuration. Then $C_0\pstep{+}C_f$ if,
and only if, $C_0\hstep{+}C_f$.
\end{proposition}
The proof is organized in the three
Lemmata~\ref{lem-ww-to-hh}--\ref{lem-commut} below.
\begin{remark}
Observe that the requirement of empty channels for $C_0$ in
\autoref{prop-psteps-hsteps} cannot be lifted, as illustrated with
\autoref{ex-write-superseding}. However, using standard coding tricks
(\textit{e.g.}, storing \emph{initial} channel contents in control states), one
can reduce a reachability or termination problem
 starting from an
arbitrary initial configuration to the same question starting from an
empty-channel $C_0$, and show its decidability by combining
\autoref{prop-psteps-hsteps} and \autoref{theo-pcs-decidability}.
\qed
\end{remark}

\begin{lemma}[From $\TS_\ww$ to $\TS_\hh$]
\label{lem-ww-to-hh}
If $\TS_\ww$ has a run $C\wstep{+}D$ then $\TS_\hh$ has a run
$C\hstep{+}D$.
\end{lemma}
\begin{proof}
We show that $\wstep{}$ is contained in $\hstep{+}$, assuming for the
sake of simplicity that $S$ has only one channel.

A writing step $(p,x)\wstep{!a}(q,y)$ with $x=z\,b_1 \cdots b_j$ and
$y=z\,a$ in $\TS_\ww$ can be simulated in $\TS_\hh$ with
$(p,x)\hstep{!a}(q,z\,b_1 \cdots b_j\,a) \hstep{\hh \ell}(q,z\,b_1\,
b_{j-1}) \hstep{\hh \ell-1} \cdots \hstep{\hh k+1} (q,z\,a)$, where
$\ell=\len{x}$ and $k=\len{z}$. Reading steps simply coincide in
$\TS_\ww$ and $\TS_\hh$.
\end{proof}

In the other direction, one can translate runs in $\TS_\hh$ to runs in
$\TS_\ww$ as stated by following lemma.
\begin{lemma}[From $\TS_\hh$ to $\TS_\ww$]
\label{lem-hh-to-ww}
If $\TS_\hh$ has a run $C\hstep{*}D$ then $\TS_\ww$ has a
run $C'\pstep{*}D$ for some $C'\hleq C$.
In particular, if the channels are empty in $C$, then necessarily
$C'=C$ and $C\pstep{*}D$.
\end{lemma}
\begin{proof}
Again we assume that $S$ has only one channel.

Write the run $C\hstep{*}D$ under the form
$C_0\hstep{}C_1\hstep{} \cdots \hstep{} C_n$ and rearrange
its steps so that superseding occurs greedily. This relies on
\autoref{lem-commut} stated next without proof.

Repeatedly applying \autoref{lem-commut} to transform
$C_0\hstep{*}C_n$ as long as possible is bound to terminate (with
each commutation, superseding steps are shifted to the left of
reliable steps, or the sum $\sum_i k_i$ of superseding positions in
steps $C_{i-1}\hstep{\hh k_i}C_i$ increases strictly while being
bounded by $O(n^2)$ for a length-$n$ run). One eventually obtains a
new run $C_0\hstep{*}C_n$ with same starting and final configurations,
and where all the superseding steps occur (at the beginning of the run
or) just after a write in \emph{normalized} sequences of the form
\begin{equation}\label{eq-norm-hstep}
C=(q,x)\hstep{! a} \hstep{\hh \ell} \hstep{\hh
  \ell-1} \hstep{\hh \ell-2}\cdots\hstep{\hh \ell-r}C'
\:,
\end{equation}
where furthermore $\ell=\len{x}$. In this case, $\TS_\ww$ has a step
$C\wstep{! a}C'$.

Greedily shifting superseding steps to the left may move some of them
at the start of the run instead of after a write: these steps are
translated into $C\hgeq C'$ in \autoref{lem-hh-to-ww}. Finally, the
steps that are not in normalized sequences are reading steps which
exist unchanged in $\TS_\ww$.
\end{proof}

\begin{lemma}[Commuting $\hh$-steps]~\label{lem-commut}\hfill
\begin{enumerate}
\item
If $C_1\hstep{?a}C_2\hstep{\hh k}C_3$ then there is a configuration
$C'_2$ s.t.\ $C_1\hstep{\hh k+1}C'_2\hstep{?a}C''$.
\item
If $C_1=(q,x)\hstep{!a}C_2\hstep{\hh k}C_3$ with $k<\len{x}$, then there is
  a configuration $C'_2$ s.t.\
$C_1\hstep{\hh k}C'_2\hstep{!a}C_3$.
\item
If $C_1=(q,x)\hstep{\hh k_1}C_2\hstep{\hh k_2}C_3$ with $k_1\leq k_2$ then there is a configuration
$C'_2$ s.t.\ $C_1\hstep{\hh k_2+1}C'_2\hstep{\hh k_1}C''$.
\end{enumerate}
\end{lemma}

\subsection{Priority Channel Systems are Well-Structured}\label{sec-wsts}
Our main result regarding the verification of PCSs is that they are
\emph{well-structured} systems, which entails the decidability of
standard decision problems via the generic decidability results
from~\cite{wsts2000,FinSch-WSTS,concur}. Let us first recall the definitions
of well-quasi-orders and well-structured systems.
\begin{definition}[wqo]\label{def-wqo}
  Let $(A,\leq_A)$ be a quasi order. Then $(A,\leq_A)$
  is a \emph{well-quasi-order} (wqo) if for any infinite sequence
  $x_0,x_1,x_2,\ldots$ over $A$ there exists two indices $i<j$ such
  that $x_i\leq_A x_j$.
\end{definition}
A simple example of a wqo is any finite set $\Sigma$ with equality
$(\Sigma,=)$, thanks to the pigeonhole principle.  More generally,
complex wqos can be build from simpler ones by algebraic
operations~\citep{SS-esslli2012}.
Let $(A_1,\leq_{A_1})$ and $(A_2,\leq_{A_2})$ be wqos:
\begin{itemize}
\item Their \emph{disjoint sum} $A_1+A_2\eqdef\{\tup{x,i}\mid
  i\in\{1,2\}\text{ and }x\in A_i\}$ is well-quasi-ordered by the
  \emph{sum ordering} $\leq_+$ defined by $\tup{x,i}\leq_+\tup{y,j}$
  $\equivdef$ $i=j$ and $x\leq_{A_i}y$.
\item Their \emph{Cartesian product} $A_1\times
  A_2\eqdef\{\tup{x,y}\mid x\in A_1\text{ and }y\in A_2\}$ is
  well-quasi-ordered by the \emph{product ordering} $\leq_\times$
  defined by $\tup{x,y}\leq_\times\tup{x',y'}$ $\equivdef$
  $x\leq_{A_1}x'$ and $y\leq_{A_2} y'$.  This is also known as
  Dickson's Lemma.
\item The Kleene star $A_1^\ast$, \textit{i.e.}, the set of finite sequences over $A_1$ is
  well-quasi-ordered by the \emph{substring embedding} relation
  $\leq_\ast$ defined by $x\leq_\ast y$ $\equivdef$ $x=a_1\cdots
  a_\ell$, $y=y_0\,b_1\,y_1\cdots y_\ell\,b_\ell\,y_{\ell+1}$ for some
  $a_i,b_i$ in $A$, some $y_i$ in $A^\ast$, and such that $a_i\leq_{A_1} b_i$ for
  every $1\leq i\leq\ell$.  This is known as Higman's Lemma, and is
  instrumental in the study of lossy channel systems (\textit{cf.}\
  \autoref{sec-related}).
\end{itemize}

\begin{definition}[WSTS]
  \label{def-wsts}
  A \emph{well-structured (transition) system} (WSTS) is a tuple
  $\TS=\mbox{$(A,\step{},\leq_A)$}$ with ${\step{}}\subseteq A\times A$
  such that
  \begin{enumerate}
  \item $(A,\leq_A)$ is a wqo; and
  \item $\step{}$ is compatible with respect to $\leq_A$,
    \textit{i.e.}, if $x\step{} y$ and $x\leq_A x'$ then there is
    some $y'$ such that $x'\step{*} y'$ and $y\leq_A y'$.
  \end{enumerate}
\end{definition}

\noindent A WSTS enjoys a stronger \emph{stuttering compatibility} if the second
condition is altered to require $x'\step{+}y'$, see~\cite[\definitionautorefname~4.4]{FinSch-WSTS}. Let
$S=(\Sigma_d,\Ch,Q,\Delta)$ be a PCS, we define the following order on
configurations of $S$: $C\leq_\hh D \equivdef$ $C$ is some
$(p,y_1,\ldots,y_m)$ and $D$ is $(p,x_1,\ldots,x_m)$ with $x_i\hleq
y_i$ for all $i=1,\ldots,m$. Equivalently, $C\leq_\hh D$ if
$C$ can be obtained from $D$ by internal superseding steps.
\begin{theorem}[PCSs are WSTSs]
  \label{theo-wsts}
  For any PCS $S$, $\TS_\hh=(\Conf_S,\step{}_\hh,\leq_\hh)$, \textit{i.e.},
  the transition system $\TS_\hh$ with configurations ordered by
  $\leq_\hh$, is a well-structured system with stuttering
  compatibility.
\end{theorem}
\begin{proof}
  We have to show that the two conditions required in
  \autoref{def-wsts} hold. Proving that
  $(\Conf_S,\leq_\hh)$ is a well-quasi-ordering is the topic
  of~\autoref{sec-embedding} and will be established in a more
  general setting in~\autoref{thm-gpe-wqo}.

  Checking stuttering compatibility is trivial with the $\leq_\hh$
  ordering. Indeed, assume that $C\leq_\hh C'$ and that $C\hstep{}D$
  is a step from the ``smaller'' configuration. Then in particular
  $C'\hstep{*}C$ by definition of $\hstep{}$, so that clearly
  $C'\hstep{+} D$ and $C'$ can simulate any step from $C$.
\end{proof}

A consequence of the well-structuredness of PCSs is the decidability
of several natural verification problems. In this paper we focus on
``Reachability,'' aka ``Safety'' when we want to check that a
configuration is \emph{not} reachable: given a PCS, an initial
configuration $C_0$, and a recursive set of configurations $G\subseteq\Conf_S$,
does $C_0\hstep{*} D$ for some $D\in G$? Another decision problem is
``Inevitability,'' \emph{i.e.}\ to decide whether all maximal runs from
$C_0$ eventually visit $G$, which includes ``Termination'' as a
special case.

\begin{theorem}%
\label{theo-pcs-decidability}
Reachability and Inevitability are decidable for PCSs.
\end{theorem}
\begin{proof}[Proof (Sketch)]
In order to apply the generic WSTS algorithms from~\cite{FinSch-WSTS},
we have to prove that the order $\leq_\hh$ is decidable, and that the
set of immediate successors of a configuration and the minimal
immediate predecessors of an upward-closed set are computable.

Deciding the ordering $\leq_\hh$ between configurations is
in \textsc{NLogSpace}; the proof of this fact is the purpose of
\autoref{rem-hleq-nlogspace}. Moreover, the operational semantics is
finitely branching and effective, \emph{i.e.}, one can compute the
immediate successors of a configuration and the minimal immediate
predecessors of an upward-closed set.

We note that Reachability and Coverability coincide (even for
zero-length runs when $C_0$ has empty channels) since $\hstep{+}$
coincides with $\geq_\hh \circ \hstep{+}$, and that the answer to a
Reachability question only depends on the (finitely many) minimal
elements of $G$. One can even compute $\Pre^*(G)$ for $G$ given,
\emph{e.g.}, as a regular subset of $\Conf_S$.

For Inevitability, the algorithms in \cite{wsts2000,FinSch-WSTS}
assume that $G$ is downward-closed but, in our case where $\hstep{+}$
and $\geq_\hh \circ \hstep{+}$ coincide, decidability can be shown for
arbitrary (recursive) $G$, as in~\cite[\theoremautorefname~4.4]{phs-rp2010}.
\end{proof}

\section{Priority Embedding}\label{sec-embedding}
In this section we establish that the superseding ordering $\hleq$ on
words enjoys the well-quasi-ordering properties we require for
reasoning about PCSs. In order to keep our results generic, as already
stated at the beginning of \autoref{sec-pcs}, we establish those
properties over an alphabet that is more general than the one
introduced in \autoref{def-pcs}. Instead of allowing for messages
consisting merely of priorities, we allow for messages over an
arbitrary well-quasi-ordering to be tagged with priority numbers.
This is in line with the algebraic operations on wqos presented at the
beginning of \autoref{sec-wsts}.
\begin{definition}[Generalized Priority Alphabet]
  Let $d\in \mathbb{N}$ be a \emph{priority level} and let
  $(\Gamma, \gammaordering)$ be a well-quasi-order, a
  \emph{generalized level-$d$ priority alphabet over $\gwqo$} is
  $\Sigma_{d,\gwqo}\eqdef \{(a,w) \mid 0\le a\le d, w\in \Gamma\}$.
\end{definition}
Subsequently, we call $\Sigma_{d,\gwqo}$ a \emph{generalized priority
  alphabet} for brevity. In analogy to the internal superseding steps
in \autoref{ssec-pcs-semantics}, we define the \emph{generalized
  priority relation $\ghstep{}$} over finite strings in
$\Sigma_{d,\gwqo}^*$ via a string rewriting system with the following
two families of rule schemata:
\begin{align}
  \label{eq-supersed-rew1}\{ (a,w)(a',w') \ghstep{} (a',w') & \mid a\le a', w\in \Gamma\}\;,\\
  \label{eq-supersed-rew2}\{ (a,w) \ghstep{} (a,w') & \mid w'\gammaordering w\}\;.
\end{align}
Informally speaking, the first line states that a string can be
rewritten if some higher-priority message supersedes a lower priority
message, and the second that any message can be rewritten to a message
that is below in the wqo $(\Gamma,\gammaordering)$.

We define ${\ghleq}\eqdef{\revhstep{\ast}}$, \emph{i.e.}, $\ghleq$ is
the reflexive transitive closure of the inverse of $\ghstep{}$. The
main purpose of this section is to prove that
$(\Sigma_{d,\gwqo}^*,{\ghleq})$ is a well-quasi-ordering, \emph{cf.}\
\autoref{def-wsts}. To this end, we will first establish a
characterization of $\ghleq$ via an embedding relation and
subsequently prove that the obtained priority embeddings yield a
well-quasi-ordering.

Before we continue, let us remark that the priority alphabet in
\autoref{def-pcs} and the relation $\hstep{}$ from \autoref{ssec-pcs-semantics} are obtained by considering
$(\Gamma,=)$ for some singleton set $\Gamma$. Whenever we drop the
index $\gwqo$, we implicitly refer to this well-quasi-ordering.
Letting $\Gamma$ be a finite set of messages represented as strings
and $\gammaordering$ the identity relation yields a generalized
priority alphabet where a priority can be assigned to each
message. Such an alphabet underlies for instance the
well-quasi-ordering that we will later use for showing that planar
planted trees are well-quasi-ordered under minors, \emph{cf.}\
\autoref{app-tree-minors}. Another example is $\Gamma=\Sigma^*$ for
some finite alphabet $\Sigma$ and where $\gammaordering$ is the
substring embedding, which allows for representing unbounded messages
on a lossy channel which are tagged with a priority level. Finally, we
wish to mention that for a generalized priority alphabet
$\Sigma_{d,\gwqo}$, if we wish to apply $\Sigma_{d,\gwqo}$ in a PCSs,
for \autoref{theo-pcs-decidability} to hold, $(\Gamma,\gammaordering)$
has to fulfill the same properties: $\gammaordering$ should be
decidable, and the set of immediate successors of a configuration and
the minimal immediate predecessors of an upward-closed set should be
computable.

\subsection{Embedding with Priorities}
\label{ssec-def-pleq}
We define here an embedding relation between finite strings over a
generalized priority alphabet $\Sigma_{d,\gwqo}^*$.  We are inspired
in this by a coarser \emph{gap embedding} relation defined by
Sch\"utte and Simpson~\cite{schutte85}, who use it to derive a
``natural'' formal statement undecidable in Peano arithmetic---the
totality of the function $H^{\varepsilon_0}$ defined later
in \autoref{sec-fg} is another well-known example of such a statement.
Besides refining the gap embedding relation to match the superseding
ordering, we also extend it to handle an underlying wqo
$(\Gamma,{\leq_\Gamma})$.

Given $x,y\in \Sigma_{d,\gwqo}^*$, we define the
\emph{generalized priority embedding $\gpleq$} by
\begin{gather*}
  x \gpleq y \: \equivdef
\left\{  \begin{array}{l}
 x \text{ is some }(a_1, v_1)\cdots (a_\ell, v_\ell)\\
  y \text{ is some } y_1 \, (a_1,w_1) \, y_2 \, (a_2,w_2)\cdots y_\ell \, (a_\ell, w_\ell)\\
  \text{such that }\forall i=1,\ldots,\ell:y_i\in \Sigma_{a_i,\gwqo}^*\text{ and
  }v_i\gammaordering w_i\;.
  \end{array}  
\right.
\end{gather*}
For example, in the singleton case, $201\pleq 22011$ but $120
\not\pleq 10210$, since factoring $10210$ as $z_1 1 z_2 2 z_3 0$ would
require $z_3=1\not\in\Sigma_0^*$. If $x\pleq y$ then $x$ is a subword
of $y$ and $x$ can be obtained from $y$ by removing factors of
messages with priority not above the first preserved message to the
right of the factor. Observe that $\gpleq$ is similar (but not
equivalent) to the Higman subword embedding for $d=0$. From the
definition above, we get the following properties which we will
implicitly use subsequently:
\begin{align}
  \epsilon \gpleq y &&\!\text{iff}\!&&  &y = \epsilon\:,
\\
\label{eq-gemb-concat}
  x_1 \gpleq y_1,\: x_2 \gpleq y_2 &&\!\text{imply}\!&&  &x_1\,x_2 \gpleq y_1\,y_2\:,
\\
\label{pleq-3}
  x_1\,x_2 \gpleq y &&\!\text{implies}\!&&  &\exists y_1 \sqsupseteq_{\pp,\gwqo} x_1
  : \exists y_2 \sqsupseteq_{\pp,\gwqo} x_2 : y=y_1\,y_2\:,
\\
\label{eq-gemb-superseding} 
  v \gammaordering w &&\!\text{implies}\!&&  &\forall 0\le a\le d:
  \forall z\in \Sigma_{a,\gwqo}^* : (a,v) \gpleq z(a,w)\:.
\end{align}

\begin{lemma}
   Let $\Sigma_{d,\gwqo}$ be a generalized priority alphabet. Then
   $(\Sigma_{d,\gwqo}^*, \gpleq)$ is a quasi-ordering.
\end{lemma}
\begin{proof}
  We have to show that $(\Sigma^*_{d,\gwqo}, \gpleq)$ is reflexive and
  transitive. Reflexivity is obvious from the definition of
  $\gpleq$. Regarding transitivity, let $x,y,z\in \Sigma_{d,\gwqo}^*$
  be such that $x\gpleq y\gpleq z$ and write
  $x=(a_1,u_1)\cdots(a_\ell,u_\ell)$. Since $x\gpleq y$, by definition
  we can write $y=y_1(a_1,v_1)\cdots y_\ell(a_\ell,v_\ell)$, where
  $u_i\gammaordering v_i$ and each $y_i=(b_{1,i},v_{1,i})\cdots
  (b_{m_i,i},v_{m_i,i})\in \Sigma_{a_i,\gwqo}^*$ for all $1\le i\le
  \ell$. Consequently, since $y\gpleq z$, we can decompose $z$ as
  $z=z_1(a_1,w_1)\cdots z_\ell(a_\ell,w_\ell)$, where each $z_i$ is of
  the form
  \begin{align*}
    z_i = z_{1,i}(b_{1,i},w_{1,i}) \cdots z_{m_i,i}(b_{m_i,i},w_{m_i,i}) z_i'.
  \end{align*}
  Since each $(b_{j,i}, w_{j,i})\in \Sigma_{a_i,\gwqo}^*$, by
  definition of $\gpleq$ we have $z_i\in \Sigma_{a_{i},\gwqo}^*$,
  hence the above decomposition of $z$ in particular yields $x\gpleq
  z$.
\end{proof}

The generalized priority embedding acts as a relational counterpart to
the more operational generalized superseding ordering. In fact
$\ghleq$ and $\gpleq$ coincide, as shown by the next lemma.
\begin{lemma}
  \label{lem-ghleq-gpleq-equivalence}
  For any $x,y\in \Sigma_{d,\gwqo}^*$, $x\ghleq y$ if, and only if, $x\gpleq y$.
\end{lemma}
\begin{proof}
  In the following, write $x$ as $x=(a_1,v_1)\cdots (a_k,v_k)$.

  Suppose $x\ghleq y$, \emph{i.e.}, $y\ghstep{\ast} x$. We show $x\gpleq y$
  by induction on the number of superseding steps. The base
  case where no superseding occurs entails $x=y$ and we rely on the reflexivity of
  $\gpleq$. For the induction step, let $y\ghstep{} z$ such that
  $z\ghstep{\ast} x$. By the induction hypothesis, $x\gpleq z$,
  \emph{i.e.}, $z$ can be factored as $z=z_1(a_1,w_1) \cdots
  z_k(a_k,w_k)$ such that $z_i\in \Sigma_{a_i,\gwqo}^*$ and
  $v_i\gammaordering w_i$ for all $1\le i\le k$. We do a case
  distinction on which rewriting rule is applied in order to obtain
  $y\ghstep{} z$:
  \begin{itemize}
  \item If $y\ghstep{} z$ via (\ref{eq-supersed-rew1}) then $y$ is
    obtained from $z$ by replacing some $z_j=z_{j,1}\cdots
    z_{j,\ell_j}$ with $z_j'=z_{j,1}\cdots z_{j,i-1} (b,w) z_{j,i}
    \cdots z_{j,\ell_j}$ for some $1\le i\le \ell_j+1$, $1\le j\le k$
    and $(b,w)$ such that in particular $b\le a_j$, and hence $z_j'\in
    \Sigma_{a_j,\gwqo}^*$. Thus $y$ factors as $y=z_1(a_1,w_1)\cdots
    z_j'(a_j,w_j)\cdots z_k(a_k,w_k)$, which allows us to conclude
    that $x\gpleq y$.

  \item If $y\ghstep{} z$ via (\ref{eq-supersed-rew2}), $y$ is
    obtained by replacing some $(a,w)$ occurring in $z$ with $(a,w')$
    for some $w'\ge_\Gamma w$. By transitivity of $\le_\Gamma$, $x\gpleq y$
    follows immediately.
  \end{itemize}

\noindent   Conversely, assume $x\gpleq y$: then $y$ factors as
  $y=y_1(a_1,w_1)\cdots y_k(a_k,w_k)$. Since for every $(a,w)$
  occurring in some $y_i$ we have $a\le a_i$, by repeatedly applying
  (\ref{eq-supersed-rew1}) we have $y\ghstep{}^\ast z=(a_1,w_1)\cdots
  (a_k,w_k)$. Moreover, $v_i\gammaordering w_i$ for all $1\le
  i\le k$, and thus by repeated application of
  (\ref{eq-supersed-rew2}) we get $z\ghstep{}^\ast x$, as required.
\end{proof}

\subsection{Priority Embedding is a Well-Quasi-Ordering}\label{ssec-refl}

The purpose of this section is to prove that $\gpleq$ is a
well-quasi-ordering. By application
of \autoref{lem-ghleq-gpleq-equivalence}, this entails that
$\ghleq$ is a well-quasi-ordering as well.
We rely for this on the algebraic operations presented in
\autoref{sec-wsts}, and on a classical tool from wqo theory,
namely \emph{order reflections}:
\begin{definition}[Order Reflection]Let $(A,\leq_A)$ and $(B,\leq_B)$
  be two quasi-orders.  An \emph{order reflection} is a mapping
  $r{:}\,A\to B$ such that $r(x)\leq_Br(y)$ implies $x\leq_Ay$.
\end{definition}
The following is folklore (and easy to see):
\begin{fact}\label{fc-refl}
  Let $(A,\leq_A)$ and $(B,\leq_B)$ be two quasi-orders and $r$ be an
  order reflection $A\to B$.  If $(B,\leq_B)$ is a wqo, then
  $(A,\leq_A)$ is a wqo.
\end{fact}

In the following, we define the \emph{height} of a sequence
$x\in\Sigma_{d,\gwqo}^*$, written $h(x)$, as being the highest
priority occurring in $x$. By convention, we let $h(\epsilon)\eqdef
-1$.  Thus, $x\in\Sigma_{h,\gwqo}^*$ if and only if $h\geq h(x)$, and
we further let $\Sigma_{-1,\gwqo}\eqdef\emptyset$. Any
$x\in\Sigma_{d,\gwqo}^*$ has a unique \emph{canonical} factorization
$x=x_0 (h,v_1) x_1 \cdots x_{m-1} (h,v_m) x_m$ where $m$ is the number
of occurrences of $h=h(x)$ in $x$ and where the $m+1$ \emph{residuals}
$x_0$, $x_1,\ldots,x_m$ are in $\Sigma_{h-1,\gwqo}^*$.

\begin{theorem}
  \label{thm-gpe-wqo}
  Let $\Sigma_{d,\gwqo}$ be a generalized priority alphabet.
  Then $(\Sigma_{d,\gwqo}^*,\gpleq)$ is a well-quasi-ordering.
\end{theorem}
\begin{proof}
  We proceed by induction on $d$. For the base case $d=-1$,
  \textit{i.e.}\ for the empty priority alphabet,
  $(\Sigma_{-1,\gwqo}^\ast,\gpleq)=(\{\epsilon\},=)$ is a wqo.

  For the induction step, a word $x\in \Sigma_{d,\gwqo}^*$ is either
  in $\Sigma_{d-1,\gwqo}^\ast$, or it has height $h(x)=d$ and then the canonical
  height factoring lets us write $x$ under the form
  \begin{align}\label{eq-fact-x}
    x = x_0 (d,v_1)x_1 \cdots x_{m-1} (d,v_m) x_m
  \end{align}
  with residuals $x_i\in \Sigma_{d-1,\gwqo}^*$ for all $0\le i\le m$. By
  the induction hypothesis, $(\Sigma_{d-1,\gwqo}^*,\gpleq)$ is a
  well-quasi-ordering. We exhibit an order reflection $r{:}\,
  \Sigma_{d,\gwqo}^*\to \Theta_{d,\gwqo}$, where
  \begin{align}\label{eq-def-theta}
     \Theta_{d,\gwqo} & \eqdef \Sigma_{d-1,\gwqo}^* + \Sigma_{d-1,\gwqo}^* \times
     ((\{d\}\times \Gamma) \times \Sigma_{d-1,\gwqo})^* \times (\{d\} \times \Gamma)
     \times \Sigma_{d-1,\gwqo}^*.
  \end{align}
  Since $\Theta_{d,\gwqo}$ is obtained from the well-quasi-orders
  $(\Sigma_{d-1,\gwqo}^*,\gpleq)$, $(\Gamma,\gammaordering)$, and
   $(\{d\},=)$ by disjoint sum, Cartesian product, and substring
  embedding, \autoref{fc-refl} will allow us to conclude that
  $(\Sigma_{d,\gwqo}^*,\gpleq)$ is a well-quasi-order.  To this
  end, for $x$ and $m$ as above, if $m>0$ we define
  \begin{equation}
    \label{eq-reflection}
    r(x) \eqdef
    \bigl(x_0,\bigl[((d,v_1),x_1)\cdots ((d,v_{m-1}),x_{m-1})\bigr], (d,v_m), x_m\bigr)\;,
  \end{equation}
  and $r(x)\eqdef x_0=x$ if $m=0$. We need to verify that, whenever
  $r(x)\preccurlyeq r(y)$ with respect to the ordering $\preccurlyeq$
  associated with $\Theta_{d,\gwqo}$ by the algebraic operations, then
  $x\gpleq y$. This is obvious when both $r(x)=x$ and $r(y)=y$ are in
  $\Sigma_{d-1,\gwqo}^\ast$. Otherwise, let $r(x)$ be as in
  \eqref{eq-reflection} and write
  \begin{align*}
    r(y) & = \bigl(y_0,\bigl[((d,w_1),y_1) \cdots ((d,w_{n-1}),y_{n-1})\bigr], (d,w_n), y_n\bigr)\;.
  \end{align*}
  Since $r(x)\preccurlyeq r(y)$, from the product ordering we obtain
  \begin{equation}\label{eq-zero-gpleq}
    x_0\gpleq y_0\;,
  \end{equation}
  $v_m\gammaordering w_n$, and $x_m\gpleq y_n$, while
  from the subword ordering we obtain the existence of indices $1\leq
  i_1,\ldots, i_{m-1}<n$ such that $v_{j} \gammaordering w_{i_j}$ and
  $x_j \gpleq y_{i_j}$ for all $0<j<m$.  Setting $i_{0}\eqdef 0$
  and $i_m\eqdef n$, observe that by \eqref{eq-gemb-concat} and
  \eqref{eq-gemb-superseding}, 
  \begin{equation}
   (d,v_j)x_j\gpleq (d,w_{i_{j-1}+1})y_{i_{j-1}+1}\cdots
   y_{i_j-1}(d,w_{i_j})y_{i_j}
  \end{equation}
  for all $0<j\leq m$, which together with \eqref{eq-zero-gpleq} and
  \eqref{eq-gemb-concat} implies $x\gpleq y$ as desired.
\end{proof}

\begin{remark}\label{rem-hleq-nlogspace}
\autoref{thm-gpe-wqo} and \autoref{lem-ghleq-gpleq-equivalence} prove
that $\hleq$ is a wqo on configurations of PCSs, as we assumed
in \autoref{sec-wsts}. There we also assumed that $\hleq$ is
decidable. We can now see that it is in \textsc{NLogSpace}, since, in
view of \autoref{lem-ghleq-gpleq-equivalence}, one can check whether $x\hleq y$
by reading $x$ and $y$ simultaneously while guessing
nondeterministically a factorization $z_1a_1\cdots z_\ell a_\ell$ of
$y$, and checking that $z_i\in\Sigma_{a_i}^*$.
\qed
\end{remark}

\section{Fast-Growing Upper Bounds}\label{sec-fg}
The verification of infinite-state systems, and WSTSs in particular,
often turns out to require astronomic computational resources expressed as
\emph{subrecursive functions}~\citep{lob70,fairtlough98} of the
input size.  We show in this section how to bound the complexity of
the algorithms presented in \autoref{sec-wsts} and classify the
Reachability and Inevitability problems for PCSs using \emph{fast-growing
complexity classes}~\citep{arXiv/Schmitz13}.

To this end, we first provide the necessary background on subrecursive
functions in \autoref{ssec-sub-hier}.  The heart of the upper bound
proof is a specialized Length Function Theorem for
$(\Sigma_{d,\Gamma}^\ast,\gpleq)$, obtained in \autoref{sub-lft} by
instrumenting the proof of \autoref{thm-gpe-wqo} and applying the
generic Length Function Theorem from \cite{SS-icalp11}.  This allows
us to derive $\F_{\ez}$ upper bounds and new \emph{combinatorial
algorithms} for PCS verification in \autoref{sub-upb}.

\subsection{Subrecursive Hierarchies}
\label{ssec-sub-hier}

Throughout this paper, we use \emph{ordinal terms} inductively defined
by the following grammar
\[
(\Omega\ni)\;\;
\alpha,\beta,\gamma ~::=~ 0 \mid \omega^\alpha \mid \alpha + \beta
\]
where addition is associative, with $0$ as the neutral element (the
empty sum).  Such a term $\alpha=\sum_{i=0}^k \omega^{\alpha_i}$ is
$0$ if $k=0$, otherwise a \emph{successor} if $\alpha_k=0$ and a
\emph{limit} otherwise.  We often write $1$ as short-hand for
$\omega^0$, and $\omega$ for $\omega^1$.  The symbol $\lambda$ is
reserved for limit ordinal terms.

We can associate a set-theoretic ordinal $o(\alpha)$ with each term
$\alpha$ by interpreting $+$ as the direct sum operator and $\omega$
as $\+N$; this gives rise to a well-founded quasi-ordering
$\alpha<\beta\equivdef o(\alpha)<o(\beta)$.  A term
$\alpha=\sum_{i=1}^k \omega^{\alpha_i}$ is in \emph{Cantor normal
  form} (CNF) if $\alpha_1\geq\alpha_2\geq\cdots\geq\alpha_k$ and each
$\alpha_i$ is itself in CNF for $i=1,\ldots,k$.  Terms in CNF and
set-theoretic ordinals below $\epsilon_0$ are in bijection; it will
however be convenient later in \autoref{sec-hardy} to consider terms
that are \emph{not} in CNF.

With any limit term $\lambda$, we associate a \emph{fundamental
  sequence} of terms $(\lambda_n)_{n\in\Nat}$:
\begin{equation}
\label{eq-def-fund-seq}
\begin{aligned}
(\gamma+\omega^{\beta+1})_n &\eqdef \gamma+\omega^\beta \cdot n = 
\gamma + \obracew{\omega^\beta + \cdots + \omega^\beta}{n}
\:,
&
(\gamma+\omega^{\lambda'})_n &\eqdef \gamma+\omega^{\lambda'_n}
\:.
\end{aligned}
\end{equation}
This yields $\lambda_0<\lambda_1 < \cdots <%
\lambda$ for any $\lambda$, with furthermore $\lambda=\lim_{n\in\Nat}
\lambda_n$.  For instance, $\omega_n=n$, $(\omega^\omega)_n=\omega^n$,
\emph{etc}.  Note that 
$\lambda_n$ is in CNF when $\lambda$ is.

We need to add a term $\ez$ to $\Omega$ to represent the set-theoretic
$\epsilon_0$, \emph{i.e.}, the smallest solution of $x=\omega^x$.  We take
this term to be a limit term as well; we define the fundamental
sequence for $\ez$ by $\ez[n]\eqdef\Omega_n$, where for $n\in\Nat$, we
use $\Omega_n$ as short-hand notation for the ordinal
$\omega^{\omega^{\cdots^\omega}}$\raisebox{0.25em}{$\bigr\}$}\raisebox{0.3em}{\footnotesize
$n$ stacked $\omega$'s}, \emph{i.e.}, for $\Omega_0\eqdef 1$ and
$\Omega_{n+1}\eqdef\omega^{\Omega_n}$.

\subsubsection{Inner Recursion Hierarchies}
Our main subrecursive hierarchy is the \emph{Hardy hierarchy}.  Given
a monotone expansive unary function $h{:}\,\+N\to\+N$, it is defined
as an ordinal-indexed hierarchy of unary functions
$(h^\alpha{:}\,\+N\to\+N)_{\alpha}$ through
\begin{align}\label{eq-def-hardy}
  h^0(n)&\eqdef n\:, &
  h^{\alpha+1}(n)&\eqdef h^\alpha\big(h(n)\big)\:,&
  h^\lambda(n)&\eqdef h^{\lambda_n}(n)\:.
\end{align}
Observe that $h^1$ is simply $h$, and more generally $h^\alpha$ is the
$\alpha$th iterate of $h$, using diagonalization to treat limit
ordinals.

A case of particular interest is to choose the successor function
$H(n)\eqdef n+1$ for $h$.  Then the \emph{fast growing hierarchy}
$(F_{\alpha})_\alpha$ can be defined by $F_\alpha\eqdef
H^{\omega^\alpha}$, resulting in $F_0(n)=H^1(n)=n+1$,
$F_1(n)=H^\omega(n)=H^{n}(n)=2n$, $F_2(n)=H^{\omega^2}(n)=2^nn$ being
exponential, $F_3=H^{\omega^3}$ being non-elementary,
$F_\omega=H^{\omega^\omega}$ being an Ackermannian function,
$F_{\omega^k}$ a $k$-Ackermannian function, and $F_{\ez}=H^{\ez}\circ
H$ a function whose totality is not provable in Peano
arithmetic~\citep{fairtlough98}.

\subsubsection{Fast-Growing Complexity Classes}
Our intention is to establish the ``$F_{\epsilon_0}$ completeness'' of
verification problems on PCSs.  In order to make this statement more
precise, we define the class $\F_{\epsilon_0}$ %
as a specific instance of the \emph{fast-growing
  complexity classes} defined for $\alpha\geq 3$ by \citep{arXiv/Schmitz13}
\begin{align}
  \F_\alpha&\eqdef\hspace{-1em}\bigcup_{p\in\bigcup_{\beta<\alpha}\FGH{\beta}}\hspace{-1em}\text{\textsc{DTime}}(F_\alpha(p(n)))\;,&
  \FGH{\alpha}&=\bigcup_{c<\omega}\text{\textsc{FDTime}}(F^c_\alpha(n))\;,
\end{align}
where the class of \emph{functions} $\FGH{\alpha}$ as defined above is the
$\alpha$th level of the \emph{extended Grzegorczyk
  hierarchy}~\citep{lob70} when $\alpha\geq 2$.

The latter hierarchy of function classes $(\FGH{\alpha})_\alpha$ is
well-established~\citep{lob70,fairtlough98}.  The class $\FGH\alpha$
is the set of functions computable in time $F^c_\alpha$, a finite
iterate of $F_\alpha$.  In particular, $\FGH{2}$ is the set of
elementary functions, $\bigcup_{\alpha<\omega}\FGH{\alpha}$ the set of
primitive-recursive functions, while
$\bigcup_{\alpha<\ez}\FGH{\alpha}$ is exactly the set of
ordinal-recursive (aka ``provably recursive'')
functions~\citep{fairtlough98}.

The complexity classes $(\F_\alpha)_\alpha$ are more
recent~\cite{arXiv/Schmitz13}: $\F_\alpha$ is the set of decision
problems that can be solved in time $F_\alpha\circ p$ for some $p$ in
$\bigcup_{\beta<\alpha}\FGH\beta$.  Each $\F_\alpha$ is naturally
equipped with $\bigcup_{\beta<\alpha}\FGH{\beta}$ as classes of
reductions.  For instance, $\FGH{2}$ is the set of elementary
functions, and $\F_3$ the class of problems whose complexity is bounded by a tower of exponentials
of height given by some elementary function of the input.\footnote{Note that, at such high complexities, the usual
distinctions between deterministic vs.\ nondeterministic, or
time-bounded vs.\ space-bounded computations become irrelevant.}

\subsection{The Length of Controlled Bad Sequences}\label{sub-lft}
A finite or infinite sequence $x_0,x_1,\dots$ over a quasi-order
$(A,\leq_A)$ is called \emph{bad} if, for all indices $i<j$,
$x_i\not\leq_A x_j$.  \autoref{def-wqo} can thus be restated by saying
that $(A,\leq_A)$ is a wqo if and only if every bad sequence over $A$
is finite.  In order to bound the complexity of the algorithms from
\autoref{theo-pcs-decidability}, we wish to bound the lengths of bad
sequences over the wqo $(\Conf_{\mathcal{S}},\leq_\#)$.  More
precisely, the main issue here is to bound the length of bad sequences
over $(\Sigma_d^\ast,\pleq)$; we actually work in the more general
case of $(\Sigma_{d,\Gamma}^\ast,\gpleq)$.

\subsubsection{Controlled Sequences}
We employ to this end the framework and results of \citep{SS-icalp11}.
The first observation is that bad sequences over
$(\Sigma_d^\ast,\pleq)$ can be of arbitrary length: for every $N>0$,
the sequence
\begin{equation*}
  1,0^N,0^{N-1},\dots,0
\end{equation*}
is indeed a bad sequence of length $N+1$ over $(\Sigma_1^\ast,\pleq)$.
Thankfully, what we are looking for are not general bounds over all
the bad sequences, but over the kind of sequences that arise in the
algorithms of \autoref{theo-pcs-decidability}: in particular, the
bounds can take into account how fast the lengths of the strings in
the sequence can grow.  Define a \emph{normed} wqo as a wqo
$(A,\leq_A)$ further equipped with a norm $|.|_A{:}\,A\to\+N$.  As a
sanity condition, we ask for
\begin{equation}
  A_{\leq n}\eqdef\{x\in A\mid |x|_A\leq n\}
\end{equation}
to be finite for each $n$.

The wqos introduced in \autoref{sec-wsts} can be normed for instance
by
\begin{align*}
  |a|_\Sigma&\eqdef 0\;,&
  |\tup{x,i}|_{A_1+A_2}&\eqdef |x|_{A_i}\;,\\
  |\tup{x,y}|_{A_1\times A_2}&\eqdef \max(|x|_{A_1},|y|_{A_2})\;,&
  |x_1\cdots x_\ell|_{A_1^\ast}&\eqdef \max_{1\leq i\leq\ell}(\ell,|x_i|_{A_1})\;,
\end{align*}
where $(\Sigma,=)$ denotes a finite set with equality.  For
$(\Sigma_{d,\Gamma}^\ast,\gpleq)$, we choose similarly
\begin{align}
  |(a_1,w_1)\cdots(a_\ell,w_\ell)|_{\Sigma_{d,\Gamma}^\ast}&\eqdef
  \max_{1\leq i\leq\ell}(\ell,|w_i|_{\Gamma})\;,
\intertext{%
where we assume $(\Gamma,\leq_\Gamma)$ to be normed by $|.|_\Gamma$.
By the definition above, this simplifies to}
  |(a_1,w_1)\cdots(a_\ell,w_\ell)|_{\Sigma^\ast_{d,\Gamma}}&=\ell
\end{align}
when $\Gamma$ is finite.

Let $g{:}\,\+N\to\+N$ be a strictly monotone function (hereafter
called a \emph{control function}), and $n$ be a non-negative integer.
A sequence $x_0,x_1,\dots$ over $(A,\leq_A,|.|_A)$
is \emph{$(g,n)$-controlled} if, for all $i$, $|x_i|_A\leq g^i(n)$, the
$i$th iterate of $g$ on $n$.  Note in particular that this entails
$|x_0|_A\leq n$.  Given an algorithm that relies on $(A,\leq_A)$ being
a wqo for its termination, \emph{i.e.}, on the fact that bad sequences over
$(A,\leq_A)$ are finite, the intuition is that $g$ should bound how
fast the norm of the elements in our bad sequences can grow, and $n$ should
bound the norm of the initial element.  As shown
in~\citep{SS-icalp11}, for a given $g$ and $n$, bad $(g,n)$-controlled
sequences over $(A,\leq_A,|.|_A)$ have a maximal length denoted
$L_{A,g}(n)$.

\subsubsection{Normed Reflections}
The \emph{Length Function Theorem} in~\citep{SS-icalp11} provides
suitable subrecursive upper bounds on the function $L_{A,g}$ when $A$
is constructed using the \emph{elementary} wqo algebra that allows
disjoint unions, Cartesian products and Kleene star to be used over
finite sets.  We are going to exploit these bounds together with the
order reflection employed in the proof of \autoref{thm-gpe-wqo} to
obtain a bound on $L_{\Sigma^\ast_{d,\Gamma},g}$.

A reflection $r{:}\,A\to B$ between two normed wqos $(A,\leq_A,|.|_A)$
and $(B,\leq_B,|.|_B)$ is \emph{normed} if $|r(x)|_B\leq |x|_A$ for
all $x$ in $A$.  We write ``$A\hookrightarrow B$'' if there exists
such a normed reflection from $A$ to $B$.  Observe that, if
$x_0,x_1,\dots$ is a $(g,n)$-controlled bad sequence over $A$, then
$r(x_0),r(x_1),\dots$ is also a $(g,n)$-controlled bad sequence, this
time over $B$.  Thus $A\hookrightarrow B$ implies $L_{A,g}\leq
L_{B,g}$, \textit{i.e.}, $L_{A,g}(m)\leq L_{B,g}(m)$ for all $m\in\Nat$.

One can check, using induction on $d$, that the reflection
$r{:}\,\Sigma_{d,\Gamma}^\ast\to\Theta_{d,\Gamma}$ used in the proof
of \autoref{thm-gpe-wqo} is normed.  If $x$ is in
$\Sigma_{d-1,\Gamma}^\ast$, then $r(x)=x$ itself with
$|x|_{\Sigma_{d,\Gamma}^\ast}=|r(x)|_{\Theta_{d,\Gamma}}$.  Otherwise, let
$x$ be factorized as in \eqref{eq-fact-x}; $r(x)$ is given
in \eqref{eq-reflection}.  Then on the one hand
\begin{align}
  |x|_{\Sigma_{d,\Gamma}^\ast}&=\max_{0\leq j<k\leq
   m}\big(m+\sum_{i=0}^m|x_i|,|v_k|_\Gamma,|x_j|_{\Sigma_{d-1,\Gamma}^\ast}\big)\;,
  \intertext{while on the other hand}
  |r(x)|_{\Theta_{d,\Gamma}}&=\max\big(|x_0|_{\Sigma_{d-1,\Gamma}^\ast},\max_{1\leq
   i\leq
   m-1}(m-1,|v_i|_\Gamma,|x_i|_{\Sigma_{d-1,\Gamma}^\ast}),|w_m|_\Gamma,|x_m|_{\Sigma^\ast_{d-1,\Gamma}}\big)\;,
   \intertext{which indeed satisfy}
   |x|_{\Sigma_{d,\Gamma}^\ast}&\geq |r(x)|_{\Theta_{d,\Gamma}}\;,
\end{align}
thanks to the ind.\ hyp. Therefore $\Sigma_{d,\Gamma}^\ast\hookrightarrow\Theta_{d,\Gamma}$ for
all $d,\Gamma$.
Define now
\begin{align}
  \Theta'_{-1,\Gamma}&\eqdef \mathbf{1}&
  \Theta'_{d,\Gamma}&\eqdef\Theta'_{d-1,\Gamma}+\Theta'_{d-1,\Gamma}\times(\Gamma\times\Theta'_{d-1,\Gamma})^\ast\times\Gamma\times\Theta'_{d-1,\Gamma}\;,
\end{align}
where $\mathbf{1}$ denotes the singleton set.
Since $\hookrightarrow$ is a precongruence for the elementary
algebraic operations~\citep[Proposition~3.5]{SS-icalp11}, we deduce
that $\Sigma_{d,\Gamma}^\ast\hookrightarrow\Theta'_{d,\Gamma}$ and thus
\begin{equation}\label{eq-maj-theta}
  L_{\Sigma_{d,\Gamma}^\ast,g}\leq L_{\Theta'_{d,\Gamma},g}
\end{equation}
for all normed wqos $\Gamma$, all $d$, and all control functions $g$.
Assuming $(\Gamma,\leq_\Gamma,|.|_\Gamma)$ to be elementary, then each
$\Theta_{d,\Gamma}'$ is also elementary, \textit{i.e.}\ we are going to be able
to apply the Length Function Theorem to it and derive an upper bound
for $L_{\Theta'_{d,\Gamma},g}$, and thereby for $L_{\Sigma_{d,\Gamma}^\ast,g}$.

\subsubsection{Maximal Order Types}
The version of the Length Function Theorem we wish to apply requires
the computation of the \emph{maximal order type} of
$\Theta'_{d,\Gamma}$.  This is a measure of the complexity of a wqo
$(A,\leq_A)$ defined in~\citet{dejongh77} as the maximal order type of
its linearizations: a \emph{linearization} $\prec$ of $\leq_A$ is a
total linear ordering over $A$ that contains $\leq_A\setminus\geq_A$
as a subrelation.  Any such linearization of a wqo is well-founded and
thus isomorphic to an ordinal, called its order type, and the maximal
order type of $(A,\leq_A)$ is therefore the maximal such ordinal.

De Jongh and Parikh~\citet{dejongh77} provide formul\ae\ to compute
the maximal order types of elementary wqos based on their algebraic
decompositions as disjoint sums, Cartesian products, and Kleene
star---using respectively the sum ordering, the product ordering, and
the subword embedding ordering---:
\begin{align*}
  o(A+B)&=o(A)\oplus o(B)                       \:,\\
  o(A\times B)&=o(A)\otimes o(B)                \:,\\
  o(A^\ast)&=\begin{cases}\omega^{\omega^{o(A)-1}}&\text{if
    $A$ is finite,}\\
    \omega^{\omega^{o(A)}}&\text{otherwise.}
\end{cases}
\end{align*}
Here, the $\oplus$ and $\otimes$ operations are the \emph{natural sum}
and \emph{natural product} on ordinals, defined for ordinals in CNF in
$\ez$ by
\begin{align}\label{eq-natural-sum}
  \sum_{i=1}^m\omega^{\beta_i}\oplus\sum_{j=1}^n\omega^{\beta'_j}&\eqdef\sum_{k=1}^{m+n}\omega^{\gamma_k}\:,&
  \sum_{i=1}^m\omega^{\beta_i}\otimes\sum_{j=1}^n\omega^{\beta'_j}&\eqdef\bigoplus_{i=1}^m\bigoplus_{j=1}^n\omega^{\beta_i\oplus\beta'_j}
\:,
\end{align}
where $\gamma_1\geq\cdots\geq\gamma_{m+n}$ is a reordering of
$\beta_1,\ldots,\beta_m,\beta'_1,\ldots,\beta'_n$.

Sch\"utte and Simpson~\cite{schutte85} compute the exact maximal order
type of a wqo related to $(\Sigma_d^\ast,\pleq)$.  Here we are content
with the maximal order type of $\Theta'_{d,\Gamma}$ (which also
provides an upper bound on  $o(\Sigma_{d,\Gamma}^\ast)$).
By~\eqref{eq-maj-theta}
and~\citep[\propositionautorefname~5.2]{SS-icalp11}, we
obtain:\footnote{To be precise,
  \citep[\propositionautorefname~5.2]{SS-icalp11} only provides bounds
  for \emph{exponential} wqos---where there are no nested applications
  of the Kleene star operation---but it can be generalized to
  elementary wqos.}
\begin{proposition}[Length Function Theorem for $\Sigma_{d,\Gamma}^\ast$]
  \label{lft-gamma}
  Let $d\in\+N$ and assume $\Gamma$ is an elementary wqo and $g$ is a
  control function.  Then there exists a polynomial $p$ independent of
  $d,\Gamma,g$ such that $L_{\Sigma_{d,\Gamma}^\ast,\,g}\leq (p\circ
  g)^{o(\Theta_{d,\Gamma}')}$.
\end{proposition}

\subsubsection{Finite Alphabets and Successor Control}
\autoref{lft-gamma} is more general than useful for deriving upper
bounds on PCSs verification.  Even if we use a generalized priority
alphabet in our PCSs, by allowing read and write rules to manipulate
pairs in $\Sigma_{d,\Gamma}$ instead of only priorities in $\Sigma_d$,
we can safely assume that the underlying alphabet $\Gamma$ is finite
and part of the input, with maximal order type $o(\Gamma)=|\Gamma|$.
Similarly, the successor function $H(x)\eqdef x+1$ can be chosen for
the control function $g$, thus $p\circ g$ in \autoref{lft-gamma} is
simply a polynomial.

We can furthermore simplify the ordinal index in \autoref{lft-gamma}.
First note that
\begin{equation}\label{eq-o-theta}
  o(\Theta'_{d,\Gamma})< (\Omega_{2(d+1)+1})_{|\Gamma|}
\end{equation}
for all $d$ in $\+N$, where $(\Omega_{2(d+1)+1})_{|\Gamma|}=\omega^{\cdots^{\omega^{|\Gamma|}}}$\raisebox{0.25em}{$\bigr\}$}\raisebox{0.3em}{\footnotesize
$2(d+1)$ stacked $\omega$'s}.
Second, an ordinal term $\alpha$ in CNF can be written
as $\alpha=\omega^{\alpha_1}\cdot c_1+\cdots+\omega^{\alpha_m}\cdot
c_m$ for $\alpha>\alpha_1>\cdots>\alpha_m$ and
$0<c_1,\dots,c_m<\omega$.  We define then inductively its \emph{maximum coefficient}
$N(\alpha)$ as $\max(N(\alpha_1),\ldots,N(\alpha_m),c_1,\ldots,c_m)$.  Observe that for
all $d$ in $\+N$,
\begin{equation}\label{eq-N-theta}
  N(o(\Theta'_{d,\Gamma}))\leq |\Gamma|\;.
\end{equation}
The following simplified statement then holds:
\begin{corollary}\label{cor-lft}
  Let $d\in\+N$ and $\Gamma$ be a finite non-empty alphabet. Then
  there exists a polynomial $h$ independent of $d,\Gamma$ such that
  $L_{\Sigma^\ast_{d,\Gamma},H}(n)\leq h^{\Omega_{2(d+1)+1}}(n)$ for
  all $n\geq |\Gamma|$.
\end{corollary}
\begin{proof}
  By \autoref{lft-gamma} it suffices to show that
  $h^{o(\Theta'_{d,\Gamma})}(n)\leq h^{\Omega_{2(d+1)+1}}(n)$ for all
  $n\geq |\Gamma|>0$.  We show instead
  $h^{o(\Theta'_{d,\Gamma})}(n)\leq
  h^{(\Omega_{2(d+1)+1})_{|\Gamma|}}(n)$ since it allows to conclude.
  By monotonicity of the Hardy functions in the ordinal index for
  the \emph{pointwise
  ordering}---see \citep[\theoremautorefname~2.21.2]{fairtlough98}
  or \citep[\lemmaautorefname~A.10]{SS-esslli2012}---it suffices to
  show that
  $o(\Theta'_{d,\Gamma})\preceq_{|\Gamma|}\left(\Omega_{2(d+1)+1}\right)_{|\Gamma|}$
  using the notation of \citep{SS-esslli2012}, which is entailed by
  (\ref{eq-o-theta}--\ref{eq-N-theta})
  and \citep[\lemmaautorefname~A.5]{SS-esslli2012}.
\end{proof}

\subsection{Complexity Upper Bounds}\label{sub-upb}
Now that we are armed with a Length Function Theorem for
$(\Sigma_{d,\Gamma}^\ast,\gpleq)$, we can prove an upper bound for PCS
verification:
\begin{theorem}[Complexity of PCS Verification]
  \label{theo-ez-upb}
  Reachability and Inevitability of PCSs are in $\F_{\ez}$.
\end{theorem}
Let us explain the steps towards an upper bound for Termination in
some detail; the results for Reachability and Inevitability are
similar but more involved---see~\citep{SS-esslli2012,concur} for
generic complexity arguments for WSTSs.

\subsubsection*{A Finite Witness}
Observe that, if an execution
$C_0\hstep{}C_1\hstep{}C_2\hstep{}\cdots$ of the transition system
$\mathcal{S}_\#$ satisfies $C_i\leq_\# C_j$ for some indices $i<j$,
then because $\mathcal{S}_\#$ is a WSTS, we can simulate the steps
performed in this sequence after $C_i$ but starting from $C_j$ and
build an infinite run.  Conversely, if the system does not terminate,
\textit{i.e.}\ if there is an infinite execution
$C_0\hstep{}C_1\hstep{}C_2\hstep{}\cdots$, then because of the wqo we
will eventually find $i<j$ such that $C_i\leq_\# C_j$.  Therefore, the
system is non-terminating if and only if there is a finite witness of
the form
$C_0\hstep{*}C_i\hstep{+}C_j$ with $C_i\leq_\# C_j$.

\subsubsection*{Controlled Witnesses}
Another observation is that the size of successive configurations
cannot grow arbitrarily along runs; in fact, the length of the
channels contents can only grow by one symbol at a time using a write
transition.  This means that if we define
$|C=(q,x_1,\dots,x_m)|=\sum_{j=1}^m|x_j|$, then in an execution
$C_0\hstep{}C_1\hstep{}C_2\hstep{}\cdots$, $|C_i|\leq
|C_0|+i=H^i(|C_0|)$, \textit{i.e.}\ any execution is \emph{controlled} by the
successor function $H$.

\subsubsection*{Applying the Length Function Theorem}
\autoref{cor-lft} yields an \mbox{$h^{\Omega_{2(d+1)+1}}(|C_0|+|\Gamma|)$}
upper bound on the length of bad $(H,|C_0|)$-controlled sequences over
\mbox{$(\Sigma_{d,\Gamma}^\ast,\gpleq)$} for some polynomial $h$.  It can be
lifted to bound the maximal length of a witness in $\mathcal{S}_\#$,
when considering instead the ordinal
$o_{S}\eqdef(\Omega_{2(d+1)+1})^m\cdot|Q|$.  Setting
$|S|=|\Delta|+|Q|+d+m+|\Gamma|$, we see that this length is less than
$H^{\ez}\left(p(|S|+|C_0|)\right)<F_{\ez}\left(p(|S|+|C_0|)\right)$
for some fixed ordinal-recursive function $p$.

\subsubsection*{A Combinatorial Algorithm}
Since the functions $(h^\alpha)_{\alpha}$ are elementary
constructive~\citep[\theoremautorefname~5.1]{arXiv/Schmitz13}, the
above discussion yields a non-deterministic algorithm in $\F_{\ez}$
for Termination: compute $L=h^{o_{S}}(|C_0|+|\Gamma|)$ and
look for an execution of length $L+1$ in $\mathcal{S}_\#$.  If one
exists, it is necessarily a witness for nontermination; otherwise, the
system is guaranteed to terminate from $C_0$.

We call this a \emph{combinatorial} algorithm, as it relies on the
combinatorial analysis provided by the Length Function Theorem to
derive an upper bound on the size of a finite witness for the property
at hand---here Termination, but the same kind of techniques can be
used for Reachability and Inevitability.

\section{Hardy Computations by PCSs}\label{sec-hardy}
In this section we show how PCSs can weakly compute the Hardy
functions $H^\alpha$ and their inverses for all ordinals $\alpha$
below $\Omega$, which is the key ingredient for
\autoref{theo-PCS-e0hard} below stating our hardness result.  For this,
we develop in \autoref{ssec-enc-ordinals} encodings
$s(\alpha)\in\Sigma_d^*$ for ordinals $\alpha\in\Omega_d$ and show how
PCSs can compute with these codes, \emph{e.g.} build the code for
$\lambda_n$ from the code of a limit $\lambda$.  This is used
in \autoref{ssec-enc-hardy} to design PCSs that
\emph{weakly compute} $H^\alpha$ and $(H^\alpha)^{-1}$ in the
sense of \autoref{def-hardy-computers} below.

\subsection{Encoding Ordinals}
\label{ssec-enc-ordinals}

Our encoding of ordinal terms as strings in $\Sigma_d^*$ employs
strings of a particular form.  For $0\leq a \leq d$, we use the
following equation to define the language $C_a\subseteq\Sigma_d^*$ of
\emph{codes}:
\begin{xalignat}{2}\label{eq-Pa}
C_{a} & \eqdef \epsilon + C_{a} C_{a-1} a
\:,
&
 C_{-1}  &\eqdef \epsilon
\:.
\end{xalignat}
Let $C=C_{-1}+C_0+\cdots +C_d$. Each $C_a$ (and then $C$ itself) is a
regular language, with $C_a=(C_{a-1}a)^\ast$; for instance,
$C_0=0^\ast$.

\subsubsection{Decompositions}
A code $x$ is either the empty word $\epsilon$, or belongs to a unique
$C_a$. If $x\in C_a$ is not empty, it has a unique factorization $x=y
z a$ according to \eqref{eq-Pa} with $y\in C_a$ and $z\in C_{a-1}$.
Recall that $h(x)$ denotes the height function, thus a non-empty
$x=a_1\cdots a_\ell$ is a code if and only if $a_\ell = h(x)$ and
$a_{i+1}-a_i \leq 1$ for all $i<\ell$ (we say that $x$ \emph{has no
  jumps}: priorities only increase smoothly along codes, but they can
decrease sharply).  For instance, $02$ is not a code (it has a jump),
but $001122$ and $01223400123334$ are codes.

The factor $z\in C_{a-1}$ in $x=y z a$ can be developed further, as
long as $z\not=\epsilon$: a non-empty code $x\in C_d$ has a unique
factorization as $x = y_d \, y_{d-1} \ldots y_a \, \std{a}$ with
$y_i\in C_i$ for $i=a,\ldots,d$, and where for $0\leq a\leq b$, we
write $\st{a}{b}$ for the \emph{staircase} word $a (a+1) \cdots (b-1)
b$, letting $\st{a}{b}=\epsilon$ when $a>b$.  We call this the
\emph{decomposition} of $x$. Note that the value of $a$ is obtained by
looking for the maximal suffix of $x$ that is a staircase word.
For example, %
 $x = 23312340121234\in C_4$ is a code %
and decomposes as
\[
x =
\obracew{2331234}{y_4} \: \obracew{\epsilon}{y_3} \: \obracew{012}{y_2} \: \obracew{\epsilon}{y_1} \: \obracew{1234}{\st{1}{4}}
\:.
\]

\subsubsection{Ordinal Encoding}
With a code $x\in C$, we associate an ordinal term $\eta(x)$ given by
\begin{xalignat}{2}\label{eq-def-eta}
\eta(\epsilon) &\eqdef 0
\:,
&
\eta(y z a) & \eqdef \eta(y)+\omega^{\eta(z)}
\:,
\end{xalignat}
where $x=y z a$ is the factorization according to \eqref{eq-Pa} of
$x\in C_a\setminus\{\epsilon\}$. For
example, $\eta(a)=\omega^0=1$ for all $a\in\Sigma_d$, $\eta(012)=\eta(234)
= \omega^\omega$, and more generally $\eta(\st{a}{b})=\Omega_{b-a}$.
One sees that $\eta(x)<\Omega_{a+1}$ when $x\in C_a$.

This \emph{decoding} function $\eta{:}\,C\to\Omega_{d+1}$ is onto (or
surjective) but it is not bijective. However, it is a bijection
between $C_a$ and $\Omega_{a+1}$ for any $a\leq d$. Its converse is
the level-$a$ \emph{encoding} function $s_a{:}\,\Omega_{a+1}\to C_a$, defined
with
\begin{xalignat}{2}
s_a\Bigl(\sum_{i=1}^p\gamma_i\Bigr) &\eqdef s_a(\gamma_1) \cdots s_a(\gamma_p)
\:,
&
s_a(\omega^\alpha) &\eqdef s_{a-1}(\alpha) \, a
\:.
\end{xalignat}
Thus $s_a(0)=\epsilon$ and, for example,
\begin{xalignat*}{3}
s_5(1) &= 5     \:,
&
s_5(3) &=  5 5 5 \:,
&
s_5(\omega) &= 4 5 \:,
\\
s_5(\omega^3) &= 4445   \:,
&
s_5(\omega^{\omega}) &=345      \:,
&
s_5(\omega^{\omega^\omega}) &=2345      \:,\\[-1.75em]
\end{xalignat*}
\begin{xalignat*}{2}
s_5(\omega^3 + \omega^2) &= 4445445     \:,
&
s_5(\omega\cdot 3) &= 4 5 4 5 4 5       \:.
\end{xalignat*}
We may omit the subscript when $a=d$, \emph{e.g.} writing
$s(1)=d$.

\subsubsection{Successors and Limits}\label{par-succ-lim}
Let $x = y_d \, y_{d-1} \ldots y_a \, \std{a}$ be the decomposition of
$x\in C_d\setminus\{\epsilon\}$. By \eqref{eq-def-eta}, $x$ encodes a
successor ordinal $\eta(x)=\beta+1$ if and only if $a=d$, \emph{i.e.}, if $x$ ends with
two $d$'s (or has length 1). Since then $\beta=\eta(y_d \ldots y_a)$,
one obtains the ``predecessor of $x$'' by removing the final $d$.

If $a<d$, $x$ encodes a limit $\lambda$. Combining
\eqref{eq-def-fund-seq} and \eqref{eq-def-eta}, one obtains the
encoding $(x)_n$ of $\lambda_n$ with
\begin{equation}
\label{eq-fund-on-codes}
(x)_n
=
y_d \, y_{d-1} \ldots y_{a+1} \bigl(y_a
(a+1)\bigr)^n \std{(a+2)}\:.
\end{equation}
For instance, with $d=5$, decomposing $x=333345=s(\omega^{\omega^4})$ gives
$a=3$, $x=y_5 y_4 y_3 \st{3}{5}$, with $y_3=333$ and
$y_5=y_4=\epsilon$. Then $(x)_n = (3334)^n 5$, agreeing with, \emph{e.g.} 
$s(\omega^{\omega^3\cdot 2})=333433345$.

\subsection{Robustness}
A crucial property of our ordinal encoding is \emph{robustness}, \emph{i.e.}
that $x\pleq x'$ should reflect the corresponding relation
$H^{\eta(x)}(n)\leq H^{\eta(x')}(n)$ on Hardy computations.
\begin{proposition}[Robustness]\label{prop-code-robust}
  Let $a\geq 0$ and $x\pleq x'$ be two strings in $C_a$.  Then,
  $H^{\eta(x)}(n)\leq H^{\eta(x')}(n')$ for all $n\leq n'$  in $\+N$.
\end{proposition}\noindent
The proof of \autoref{prop-code-robust} requires delving in some of
the theory of subrecursive functions, and is postponed until
\autoref{ssec-robust}.

\subsubsection{Properties of the Hardy Hierarchy}
We first list some useful properties of Hardy computations (see
\citep{fairtlough98} or \citep[App.~A]{SS-esslli2012} for details).
The first fact is that each Hardy function is expansive and monotone
in its argument $n$:
\begin{fact}[Expansiveness and Monotonicity]For all $\alpha,\alpha'$ in $\Omega$ and $n>0,m$
  in $\+N$,
\begin{gather}
    \label{eq-exp-hardy}
    n\leq H^\alpha(n)\:,\\
    \label{eq-mono-hardy}
    n\leq m\text{ implies }H^\alpha(n)\leq H^\alpha(m)\:.
\end{gather}
\end{fact}
However, the Hardy functions are not monotone in the ordinal
parameter: $H^{n+1}(n)=2n+1>2n=H^{n}(n)=H^\omega(n)$, though
$n+1<\omega$.  We will introduce an ordering on ordinal terms in
\autoref{sub:embd} that ensures
monotonicity of the Hardy functions.

Another handful fact is that we can decompose Hardy computations:
\begin{fact}
  For all $\alpha,\gamma$ in $\Omega$, and $n$ in $\+N$,
  \begin{equation}\label{lem-hardy-sum}
    H^{\gamma+\alpha}(n)=H^\gamma\!\left(H^\alpha(n)\right).
  \end{equation}
\end{fact}\noindent
Note that \eqref{lem-hardy-sum} holds for all ordinal terms, and not
only for those $\alpha,\gamma$ such that $\gamma+\alpha$ is in
CNF---this is a virtue of working with terms rather than set-theoretic
ordinals.

\subsubsection{Ordinal Embedding}
\label{sub:embd}
We introduce a partial ordering $\leqo$ on ordinal terms, called
\emph{embedding}, and which corresponds to a strict tree embedding
on the structure of ordinal terms.  Formally, it is defined inductively by
$\alpha\leqo\beta$ $\equivdef$
$\alpha=\omega^{\alpha_1}+\cdots+\omega^{\alpha_p}$,
$\beta=\omega^{\beta_1}+\cdots+\omega^{\beta_m}$, and there exist
$1\leq i_1<i_2<\ldots<i_p\leq m$ such that
$\alpha_1\leqo\beta_{i_1}\wedge\cdots\wedge\alpha_p\leqo\beta_{i_p}$.
Note that $0\leqo \alpha$ for all $\alpha$, that $1\leqo\alpha$ for
all $\alpha>0$.  In general, $\alpha\not\leqo\omega^\alpha$ and
$\lambda_n\not\leqo\lambda$.  Ordinal embedding is congruent for addition
and $\omega$-exponentiation of terms:
\begin{gather}
  \label{eq-leqo-congs}
  \alpha\leqo\alpha'\text{ and }\beta\leqo\beta'\text{ imply
  }\alpha+\beta\leqo\alpha'+\beta'\:,\\
  \label{eq-leqo-conge}
  \alpha\leqo\alpha'\text{ implies
  }\omega^\alpha\leqo\omega^{\alpha'}\:,
\end{gather}
and could in fact be defined alternatively by the axiom $0\leqo\alpha$
and the two deduction rules \eqref{eq-leqo-congs} and
\eqref{eq-leqo-conge}.

We list a few useful consequences of the definition of $\leqo$:
\begin{align}
\alpha\leqo \gamma+\omega^\beta
&\text{ implies }\alpha\leqo\gamma,\label{eq-D1}
		 \text{ or } \alpha=\gamma'+\omega^{\beta'} \text{ with } \gamma'\leqo\gamma \text{ and } \beta'\leqo\beta
\:,
\\
\label{eq-D2}
n\leq m
&\text{ implies } \lambda_n\leqo\lambda_m
\:,
\\
\label{eq-D3}
\alpha\leqo \lambda
&\text{ implies } \alpha\leqo\lambda_n,
\text{ or } \alpha \text{ is a limit and } \alpha_n\leqo\lambda_n
\:.
\end{align}
\begin{proof}[Proof of \eqref{eq-D1}]
Intuitively, there are two cases when we consider
$\alpha\leqo\alpha'=\gamma+\omega^\beta$: either the $\omega^\beta$
summand of $\alpha'$ is in the range of the embedding or not. If it is not,
then already $\alpha\leqo\gamma$. If it is, then $\alpha$ must be some
$\gamma'+\omega^{\beta'}$ and $\omega^{\beta'}\leqo\omega^\beta$,
which implies in turn $\beta'\leqo\beta$.
\end{proof}
\begin{proof}[Proof of \eqref{eq-D2}]
By induction on $\lambda$: indeed if $\lambda=\gamma+\omega^{\beta+1}$
then $\lambda_m=\gamma+\omega^\beta\cdot m$, %
which is $\lambda_n+\omega^\beta\cdot(m-n)$. If
$\lambda=\gamma+\omega^{\lambda'}$, the ind.\ hyp.\ gives
$\lambda'_n\leqo\lambda'_m$, hence
$\lambda_n=\gamma+\omega^{\lambda'_n}\leqo\gamma+\omega^{\lambda'_m}=\lambda_m$.
\end{proof}
\begin{proof}[Proof of \eqref{eq-D3}]
By induction on $\lambda$. We can write $\lambda$ as some
$\gamma+\omega^{\beta}$ with $\beta>0$ so that
$\lambda_n=\gamma+(\omega^\beta)_n$. If $\alpha\leqo\gamma$, then
$\alpha\leqo\lambda_n$ trivially. If $\alpha=\gamma'+ 1$ is a
successor, $1\leqo(\omega^\beta)_n$ and again
$\alpha\leqo\lambda_n$. There remains the case where
$\alpha=\gamma'+\omega^{\beta'}$ is a limit (\emph{i.e.} $\beta'>0$) with
$\gamma'\leqo\gamma$ and $\beta'\leqo\beta$. If $\beta$ is a limit,
then by ind.\ hyp.\ either $\beta'\leqo\beta_n$ and hence
$\alpha\leqo\lambda_n$, or $\beta'$ is a limit and
$\beta'_n\leqo\beta_n$, hence $\alpha_n\leqo\lambda_n$. Finally, if
$\beta=\delta+1$ is a successor, then either $\beta'\leqo\delta$ so
that $\alpha\leqo\gamma+\omega^\delta\leqo \gamma+\omega^\delta\cdot
n=\lambda_n$, otherwise by \eqref{eq-D1}, $\beta'$ is a successor
$\delta'+1$ with $\delta'\leqo\delta$, and then
$(\omega^{\beta'})_n=\omega^{\delta'}\cdot n\leqo \omega^{\delta}\cdot
n=(\omega^{\beta})_n$, hence $\alpha_n\leqo\lambda_n$.
\end{proof}
\begin{proposition}[Monotonicity]
\label{prop-leqo-mono}
For all $\alpha,\alpha'$ in $\Omega$ and $n$ in $\+N$,
\begin{equation*}
\alpha\leqo\alpha'\text{ implies }H^\alpha(n)\leq H^{\alpha'}(n)\:.
\end{equation*}
\end{proposition}
\begin{proof}
  Let us proceed by induction on a proof of $\alpha\leqo\alpha'$, based
  on the deduction rules \eqref{eq-leqo-congs} and
  \eqref{eq-leqo-conge}.  For the base case, $0\leqo\alpha'$ implies
  $H^0(n)=n\leq H^{\alpha'}(n)$ by expansiveness.

  For the inductive step with
  \eqref{eq-leqo-congs}, if $\alpha\leqo\alpha'$ and
  $\beta\leqo\beta'$, then
  \begin{align*}
    H^{\alpha+\beta}(n)
    &=H^\alpha\bigl(H^\beta(n)\bigr)&&\text{by \eqref{lem-hardy-sum}}\\
    &\leq H^{\alpha}\bigl(H^{\beta'}(n)\bigr)&&\text{by ind.\ hyp.\ and
      \eqref{eq-mono-hardy}}\\
    &\leq H^{\alpha'}\bigl(H^{\beta'}(n)\bigr)&&\text{by ind.\ hyp.}\\
    &=H^{\alpha'+\beta'}(n)\;.&&\text{by \eqref{lem-hardy-sum}}
  \end{align*}

  For the inductive step with \eqref{eq-leqo-conge}, if
  $\alpha\leqo\alpha'$, then we show $H^{\omega^\alpha}(n)\leq
  H^{\omega^{\alpha'}}(n)$ by induction on $\alpha'$:
  \begin{itemize}
  \item If $\alpha'=0$, then $\alpha=0$ and we are done.
  \item If $\alpha'=\beta'+ 1$ is a successor, then by
    \eqref{eq-D1} either $\alpha\leqo\beta'$, or $\alpha=\beta+ 1$
    with $\beta\leqo\beta'$.
    In the first case, $H^{\omega^\alpha}(n)\leq
    H^{\omega^{\beta'}}(n)\leq
    H^{\omega^{\beta'}}(H(n))=H^{\omega^{\alpha'}}(n)$ by ind.\ hyp.\ and expansiveness.
    In the second case, we see by induction on $i\in\+N$ that
    \begin{equation}\label{eq-Homi}
      \left(H^{\omega^\beta}\right)^i(n)\leq\left(H^{\omega^{\beta'}}\right)^i(n)
    \end{equation}
    for all $i$ and $n$ thanks to the ind.\ hyp.  Thus
    \begin{equation*}
      H^{\omega^{\beta+1}}(n)=\left(H^{\omega^{\beta}}\right)^n(n)\leq\left(H^{\omega^{\beta'}}\right)^n(n)=H^{\omega^{\beta'+1}}(n)
    \end{equation*}
    for all $n$, and we are done.
  \item If $\alpha'=\lambda'$ is a limit, then by \eqref{eq-D3} either
    $\alpha\leqo\lambda'_n$ or $\alpha$ is a limit $\lambda$ and $\lambda_n\leqo\lambda'_n$.  In the first case $H^{\omega^\alpha}(n)\leq
    H^{\omega^{\lambda_n}}(n)$ by ind.\ hyp.; in the second case
    $H^{\omega^\lambda}(n)=H^{\omega^{\lambda_n}}(n)\leq
    H^{\omega^{\lambda'_n}}(n)=H^{\omega^{\lambda'}}(n)$ using the
    ind.\ hyp.\qedhere
  \end{itemize}
\end{proof}

\subsubsection{Robustness}\label{ssec-robust}
We are now in position to prove \autoref{prop-code-robust}:
\begin{proof}[Proof of \autoref{prop-code-robust}]
  We prove that $\eta(x)\leqo\eta(x')$ by induction on $x$ and
  conclude using \autoref{prop-leqo-mono} and
  Eq.~\eqref{eq-mono-hardy}. If $x=\epsilon$,
  $\eta(x)=0\leqo\eta(x')$.  Otherwise we can decompose $x$ as $yza$
  according to \eqref{eq-Pa} with $y\in C_a$ and $z\in C_{a-1}$.  By
  \eqref{pleq-3}, $x'=y'z'a$ with $y\pleq y'$ and $za\pleq z'a$.
  Observe that $y'$ and $z'$ are in $C_a$, and writing
  $z'a=z'_1a\cdots z'_ma$ for the canonical decomposition of
  $z'$---where necessarily each $z'_j$ is in $C_{a-1}$---, then
  $z\pleq z'_1$ as there is no other way of disposing of the other
  occurrences of $a$ in $z'$.

  By ind.\ hyp., $\eta(y)\leqo\eta(y')$ and $\eta(z)\leqo\eta(z'_1)$.
  Then, because $\eta(x)=\eta(y)+\omega^{\eta(z)}$ and
  $\eta(x')=\eta(y')+\omega^{\eta(z'_1)}+\cdots+\omega^{\eta(z'_m)}$,
  we see by \eqref{eq-leqo-congs} and \eqref{eq-leqo-conge} that
  $\eta(x)\leqo\eta(x')$.
\end{proof}

\subsection{Robust Hardy Computations in PCSs}
\label{ssec-enc-hardy}
Our goal is to implement in a PCS the \emph{canonical} Hardy
steps, denoted with $\step{H}$, and specified by the following two
rewrite rules on pairs in $\Omega\times\+N$:
\begin{xalignat}{2}
\label{eq-hstep-succ}
(\alpha+1,n) & \step{H} (\alpha,n+1) && \text{for successors,}\\
\label{eq-hstep-limit}
(\lambda,n) &\step{H} (\lambda_n,n)  &&\text{for limits.}
\end{xalignat}
A \emph{Hardy computation} for $H^\alpha(n)$ is a sequence of rewrites
$(\alpha,n)=(\alpha_0,n_0)\step{H}(\alpha_1,n_1)\step{H}\cdots\step{H}(\alpha_\ell,n_\ell)$.
Note that (\ref{eq-hstep-succ}--\ref{eq-hstep-limit}) provide a
rewriting view of the definition of Hardy functions
in~\eqref{eq-def-hardy}.  Thus $H^{\alpha_i}(n_i)$ remains constant
throughout the computation, and if $\alpha_\ell=0$ then
$n_\ell=H^\alpha(n)$, in which case we call the computation
\emph{complete}.
\begin{figure}[tbp]
\centering
\begin{tikzpicture}[node distance=.4em,text depth=.2ex,text height=1.5ex]
  \node(o){$\tto\,$:};
  \node[right=of o,fill=black!12,text width=7em,align=flush right]
    (oc){$3\,3\,4\,5\,4\,5\,\dol\;$};
  \node[right=of oc]
    (oe){the \emph{ordinal} term $\omega^{\omega^2}+\omega^\omega$};
  \node[below=of o](c){$\ttc\,$:};
  \node[right=of c,fill=black!12,text width=7em,align=flush right]
    (cc){$0\,0\,0\,0\,\dol\;$};
  \node[right=of cc]
    (ce){the \emph{counter} value $4$};
  \node[below=of c](t){$\ttt\,$:};
  \node[right=of t,fill=black!12,text width=7em,align=flush right]
    (tc){$\dol\;$};
  \node[right=of tc]
    (te){the \emph{temporary} storage};
  \draw[thick,color=black!80]
    (oc.north west) -- (oc.north east)
    (oc.south west) -- (oc.south east)
    (cc.north west) -- (cc.north east)
    (cc.south west) -- (cc.south east)
    (tc.north west) -- (tc.north east)
    (tc.south west) -- (tc.south east);
\end{tikzpicture}
\caption{\label{fig-pcs-channels}Channels for Hardy computations.}
\end{figure}

We do not implement canonical Hardy steps as PCSs, but construct
instead \emph{weak} computers, which might return lower values.
Our PCSs for weak Hardy computations use three channels
(see \autoref{fig-pcs-channels}), storing (codes for) a pair
$\alpha,n$ on channels $\tto$ (for ``ordinal'') and $\ttc$ (for ``counter''),
and employ an extra channel, $\ttt$, for ``temporary'' storage.
Instead of $\Sigma_d$, we use $\Sigma_{d+1}$ with $d+1$ used
as a position marker and written $\dol$ for clarity: each channel
always contains a single occurrence of $\dol$.
\begin{definition}\label{def-hardy-computers}
A \emph{weak Hardy computer} for $\Omega_{d+1}$ is a $(d+1)$-PCS $S$ with channels
$\Ch=\{\tto,\ttc,\ttt\}$ and two distinguished states
$p_\init$ and $p_\final$ such that:
\begin{align}
\tag{safety}
\text{if } (p_\init,x\dol,y\dol,z\dol) &\hstep{*} (p_\final,u,v,w)
&\text{then }&x\in C_d, y\in 0^+, z=\epsilon, u,v,w\in\Sigma_d^*\dol
\:,\\
\tag{robustness}
\text{if } (p_\init,s(\alpha)\dol,0^n\dol,\dol) &\hstep{*} (p_\final,s(\beta)\dol,0^m\dol,\dol)
&\text{then }&H^\alpha(n)\geq H^\beta(m)
\:.
\intertext{%
Furthermore $S$ is \emph{complete} if for any $\alpha<\Omega_{d+1}$ and $n>0$,}
(p_\init,s(\alpha)\dol,0^n\dol,\dol)&\hstep{*}
(p_\final,\dol,0^m\dol,\dol)&\text{where }&m\eqdef H^\alpha(n) \tag{complete}
\intertext{and it is \emph{inv-complete}
if} %
(p_\init,\dol,0^m\dol,\dol)&\hstep{*}
(p_\final,s(\alpha)\dol,0^n\dol,\dol)\:.\tag{inv-complete}
\end{align}
\end{definition}

\begin{lemma}[PCSs weakly compute Hardy functions]\label{thm-hardy-computers}
For every $d\in\Nat$, there exists a weak Hardy computer $S_d$ for
$\Omega_{d+1}$ that
is complete, and a weak $S^{-1}_d$ that is inv-complete. Furthermore
$S_d$ and $S^{-1}_d$ can be generated uniformly in \textsc{DSpace}$(\log d)$.
\end{lemma}

We design a complete weak Hardy computer for
\autoref{thm-hardy-computers} by assembling several components.
The weak Hardy computer $S_d$ is actually composed of two
components $S_{d,+1}$ and $S_{d,\lambda}$ in charge respectively of
applying the successor~\eqref{eq-hstep-succ} and
limit~\eqref{eq-hstep-limit} steps, and similarly $S^{-1}_d$ is composed
of two components $S^{-1}_{d,+1}$ and $S^{-1}_{d,\lambda}$ in charge of
reversing those steps.

\subsubsection{Successor Steps}
We start with ``canonical successor steps'', as \textit{per} \eqref{eq-hstep-succ}.
They are implemented by $S_{d,+1}$, the PCS
depicted in \autoref{fig-pcs-succ}.
When working on codes, replacing $s(\alpha+1)$ by $s(\alpha)$ simply
means removing the final $d$ (see \autoref{par-succ-lim}), but when
the strings are in fifo channels this requires reading the whole
contents of a channel and writing them back, relying on the $\dol$
end-marker.
\begin{figure}[tbhp]
\centering
\scalebox{1.0}{
  \begin{tikzpicture}[->,>=stealth',shorten >=1pt,node distance=4cm,thick,auto,bend angle=30]
  \tikzstyle{every state}=[fill=white,draw=black,text=black]
\path
(0,0)   node [shape=circle,fill=black!12,draw=black!80] (p1) {$p$}
(2.3,0) node [shape=circle,fill=black!12,draw=black!80] (p3) {}
(3.5,0) node [shape=circle,fill=black!12,draw=black!80] (p4) {}
(5.0,0) node [shape=circle,fill=black!12,draw=black!80] (p5) {$q$}
(6.3,0) node [shape=circle,fill=black!12,draw=black!80] (p6) {}
(7.8,0) node [shape=circle,fill=black!12,draw=black!80] (p7)  {$r$}
;

\path (p1) edge node {$\tto\,\qb x\in C_d$} (p3);
\path (p3) edge node {$\tto \, ? \, d$} (p4);
\path (p4) edge node {$\tto \, \qb \, \dol $} (p5) ;
\path (p5) edge [loop above] node {$\ttc \, \qb \, 0 $} (p5) ;
\path (p5) edge node {$\ttc \, ! \, 0$} (p6);
\path (p6) edge node {$\ttc \, \qb \, \dol$} (p7);

  \end{tikzpicture}
}%
\caption{$S_{d,+1}$, a PCS for Hardy steps $(\alpha+1,n)\step{H}(\alpha,n+1)$.}
\label{fig-pcs-succ}
\end{figure}
\begin{remark}[Notational/graphical conventions]\label{rem-shorthand-rules}
The label edge \mbox{``$q\!\step{\ttc\qb 0}\!q$''} in
\autoref{fig-pcs-succ}, with $\ttc\qb 0$ as label, is shorthand notation for
\mbox{``$q\!\step{\ttc?0}\!\circ\!\step{\ttc!0}\!q$''}, letting 
the intermediary state remain implicit.
We also use meta-rules like ``$p\!\step{\tto\,\qb\,
x\in C_d\,}\!\circ$'' above to denote a subsystem tasked with reading and
writing back a string $x$ over $\tto$ while checking that it belongs
to $C_d$; since $C_d$ is a regular language, such subsystems are
directly obtained from DFAs for $C_d$.
\qed
\end{remark}

We first analyze the behavior of $S_{d,+1}$ when superseding of
low-priority messages does not occur, \textit{i.e.}, we first consider its
``reliable'' semantics. In this case, starting $S_{d,+1}$ in state $p$
performs the step given in \eqref{eq-hstep-succ} for successor
ordinals. More precisely, $S_{d,+1}$ guarantees
\begin{equation}
\label{eq-s1-reliable}
(p,\: s(\alpha+1)\dol,\: 0^n\dol,\: \dol)\rstep{*} (r,u,v,w)
\text{ iff }%
 u=s(\alpha)\dol \;\land\; v=0^{n+1}\dol \;\land\; w=\dol
\:.
\end{equation}
Note that \eqref{eq-s1-reliable} refers to ``$\rstep{*}$'', with no superseding. %

Observe that $S_{d,+1}$ has to non-deterministically guess where the end of
$s(\alpha)$ occurs before reading $d\dol$ in channel $\tto$, and will deadlock if it guesses
incorrectly. 
We often rely on this kind of non-deterministic
programming to reduce the size of the PCSs we build. Finally, we
observe that if $x$ does not end with $dd$ (and is not just $d$), \textit{i.e.},
if $\eta(x)$ is not a successor ordinal, then $S_{d,+1}$ will certainly deadlock.

\begin{figure}[tbhp]
\centering
\scalebox{1.0}{
  \begin{tikzpicture}[->,>=stealth',shorten >=1pt,node distance=4cm,thick,auto,bend angle=30]
  \tikzstyle{every state}=[fill=white,draw=black,text=black]
\path
(0,0)	node [shape=circle,fill=black!12,draw=black!80] (p1) {$p$}
(1.4,0) node [shape=circle,fill=black!12,draw=black!80] (p2) {}
(3.0,0) node [shape=circle,fill=black!12,draw=black!80] (p3) {}
(5.0,0) node [shape=circle,fill=black!12,draw=black!80] (p4) {}
(6.4,0) node [shape=circle,fill=black!12,draw=black!80] (p5) {}
(8.0,0) node [shape=circle,fill=black!12,draw=black!80] (p6) {$r$}
;

\path (p1) edge node {$\ttc \,?\, 0$} (p2);
\path (p2) edge node {$\ttc \,\qb\, 0^n\dol$} (p3);
\path (p3) edge node {$\tto \,\qb\, x \in C_d$} (p4);
\path (p4) edge node {$\tto \,!\, d$} (p5);
\path (p5) edge node {$\tto \,\qb\, \dol$} (p6);

  \end{tikzpicture}
}%
\caption{$S^{-1}_{d,+1}$, a PCS for inverse Hardy steps $(\alpha,n+1)\step{H^{-1}}(\alpha+1,n)$.}
\label{fig-pcs-invsucc}
\end{figure}
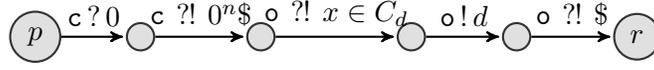
We  now  consider $S^{-1}_{d,+1}$, the PCS depicted in \autoref{fig-pcs-invsucc}
that implements the inverse canonical steps
$(\alpha,n+1)\step{H^{-1}}(\alpha+1,n)$. Implementing such steps on
codes is an easy string-rewriting task since $s(\alpha+1)=s(\alpha)d$,
however our PCS must again read the whole contents of its channels,
write them back with only minor modifications while fulfilling the
safety requirement of \autoref{def-hardy-computers}.
When considering the reliable behavior, $S^{-1}_{d,+1}$ guarantees
\begin{equation}
\label{eq-s2-reliable}
(p,\: x\dol,\: y\dol,\: \dol)\rstep{*}(r,u,v,w)
\text{ iff }
x\in C_d, \exists n:y=0^{n+1},
u=s(\eta(x)+1)\dol, v=0^n\dol, \text{ and } w=\dol
\:.
\end{equation}

Consider now the behavior of $S_{d,+1}$ and $S^{-1}_{d,+1}$ \emph{when superseding may occur}.
Note that a run $(p,x\dol,y\dol,z\dol)\wstep{*}(r,\ldots)$ from $p$ to
$r$ is a \emph{single-pass} run: it reads the whole contents of
channels $\tto$ and $\ttc$ once, and writes some new contents. This
feature assumes that we start with a single $\dol$ at the end of each
channel, as expected by $S_{d,+1}$. For such single-pass runs, the PCS
behavior with superseding semantics can be derived from the reliable
behavior: for single-pass runs, $C\wstep{*}D$ if and only if $C\rstep{*}D'\hgeq
D$ for some $D'$.

Combined with \eqref{eq-s1-reliable}, the above remark entails
robustness for $S_{d,+1}$: there is an execution
$(p,s(\alpha+1)\dol,0^n\dol,\dol)\wstep{*}
(r,s(\beta)\dol,0^{n'}\dol,\dol)$ if and only if $s(\beta)\hleq
s(\alpha)$ and $0^{n'}\dol\hleq 0^{n+1}\dol$, \textit{i.e.}, $n'\leq
n+1$. With \autoref{prop-code-robust}, we deduce $H^\beta(n')\leq
H^\alpha(n)$.

The same reasoning applies to $S^{-1}_{d,+1}$ since this PCS also performs
single-pass runs from $p$ to $r$, hence
$(p,s(\alpha)\dol,0^{n}\dol,\dol)\wstep{*}
(r,s(\beta)\dol,0^{n'}\dol,\dol)$ if and only if $s(\beta)\hleq s(\alpha+1)$ and
$n'\leq n-1$. Thus $H^\beta(n')\leq H^\alpha(n)$.

\subsubsection{Limit Steps}
Our next component is $S_{d,\lambda}$, see \autoref{fig-pcs-lim}, which
implements the canonical Hardy steps for limits
from \eqref{eq-hstep-limit}.
The construction follows \eqref{eq-fund-on-codes}: $S_{d,\lambda,a}$ reads
(and writes back) the contents of channel $\tto$; guessing
non-deterministically the decomposition $y_d\ldots y_{a+1}y_a a\std{(a+1)}$
of $s(\lambda)$, it writes back $y_d\ldots y_{a+1}$ and copies $y_a$
on the temporary $\ttt$ with $a+1$ appended. Then, a loop around
state $q_a$ copies $0^n$ from and back to $\ttc$. Every time one $0$
is transferred, the whole contents of $\ttt$, initialized with $y_a(a+1)$, is copied to
$\tto$. When the loop has been visited $n$ times, $S_{d,\lambda,a}$ empties
$\ttt$ and resumes
the transfer of $s(\lambda)$ by copying the final $\std{(a+2)}$.

For clarity, $S_{d,\lambda,a}$ as given in \autoref{fig-pcs-lim} assumes that
$a$ is fixed. The actual $S_{d,\lambda}$ component %
guesses non-deterministically
what is the value of $a$ for the $s(\lambda)$  code on
$\tto$ and gives the control to $S_{d,\lambda,a}$ accordingly.
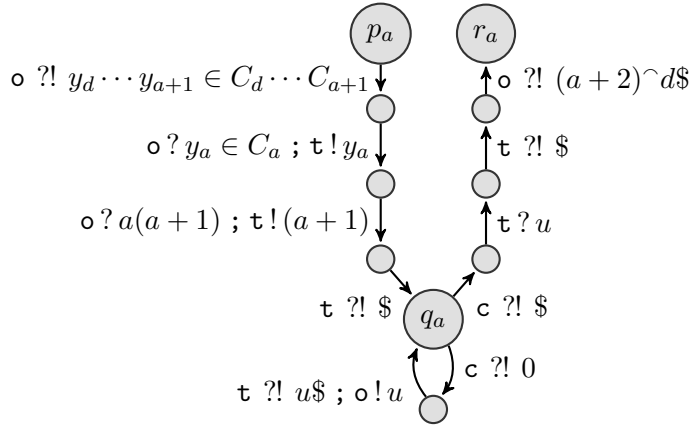
\begin{figure}[tbhp]
\centering
\scalebox{1.0}{
  \begin{tikzpicture}[->,>=stealth',shorten >=1pt,node distance=4cm,thick,auto,bend angle=30]
  \tikzstyle{every state}=[fill=white,draw=black,text=black]
\path
(0,0)   node [shape=circle,fill=black!12,draw=black!80] (p1) {$p_a$}
(0,-1) node [shape=circle,fill=black!12,draw=black!80] (p2) {}
(0,-2) node [shape=circle,fill=black!12,draw=black!80] (p3) {}
(0,-3) node [shape=circle,fill=black!12,draw=black!80] (p4) {}
(0.7,-3.8) node [shape=circle,fill=black!12,draw=black!80] (p5) {$q_a$}
(0.7,-5) node [shape=circle,fill=black!12,draw=black!80] (p5l) {}
(1.4,-3) node [shape=circle,fill=black!12,draw=black!80] (p6) {}
(1.4,-2) node [shape=circle,fill=black!12,draw=black!80] (p7) {}
(1.4,-1) node [shape=circle,fill=black!12,draw=black!80] (p8) {}
(1.4,-0) node [shape=circle,fill=black!12,draw=black!80] (p9) {$r_a$}
;

\path (p1) edge node [swap] {$\tto \,\qb\, y_d\cdots y_{a+1}\in C_d\cdots C_{a+1}$} (p2);
\path (p2) edge node [swap] {$\tto \,?\, y_a\in C_a \:\semicolon\: \ttt \,!\, y_a$} (p3);
\path (p3) edge node [swap] {$\tto \,?\, a(a+1) \:\semicolon\: \ttt \,!\, (a+1)$} (p4);
\path (p4) edge node [swap,pos=0.55] {$\ttt \,\qb\, \dol$} (p5);
\path (p5) edge node [swap] {$\ttc \,\qb\, \dol$} (p6);
\path (p6) edge node [swap] {$\ttt \,?\, u$} (p7);
\path (p7) edge node [swap] {$\ttt \,\qb\, \dol$} (p8);
\path (p8) edge node [swap] {$\tto \,\qb\, \std{(a+2)}\dol$} (p9);
\path (p5) edge [bend left] node [pos=0.40] {$\ttc \,\qb\, 0$} (p5l);
\path (p5l) edge [bend left] node [pos=0.572] {$\ttt \,\qb\, u\dol\:\semicolon\:\tto \,!\, u$} (p5);

  \end{tikzpicture}
}%
\caption{$S_{d,\lambda,a}$, a PCS for Hardy steps $(\lambda,n)\step{H}(\lambda_n,n)$.}
\label{fig-pcs-lim}
\end{figure}

As far as reliable steps are considered, $S_{d,\lambda}$ guarantees
\begin{equation}
\label{eq-s3-reliable}
(p,\: s(\alpha)\dol,\: 0^n\dol,\: \dol)\rstep{*} (r,u,v,w)
\text{ iff }
\alpha \in\LIM, u=s(\alpha_n)\dol, v=0^{n}\dol, \text{ and } w=\dol
\:.
\end{equation}
If superseding is allowed, a run
$(p_a,s(\alpha)\dol,0^n\dol,\dol)\wstep{*}(r_a,u,v,w)$
has the form
\begin{align*}
(p_a,s(\alpha)\dol,0^n\dol,\dol)
&\wstep{*}
C_0=(q_a,\std{(a+2)}\dol x_0,0^n\dol,z_0\dol)
\\
&\wstep{*}
C_1=(q_a,\std{(a+2)}\dol x_1,0^{n-1}\dol v_1,z_1\dol)\\
&\quad\vdots
\\
&\wstep{*}
C_n=(q_a,\std{(a+2)}\dol x_n,\dol v_n,z_n\dol)
\\&\wstep{*}
(r_a,x'_n\dol,v'_n\dol,\dol)
\end{align*}
where $C_i=(q_a,\std{(a+2)}\dol x_i,0^{n-i}\dol v_i,z_i\dol)$ occurs
when state $q_a$ is visited for the $i$th time. Since the run is
single-pass on $\ttc$, we know that $v_i\hleq 0^i$ for all
$i=0,\ldots,n$. Since it is single-pass on $\tto$, we deduce that
$x_0\hleq y_d\ldots y_{a+1}$, then $x_{i+1}\hleq x_i z_i$ for all $i$,
and finally $x'_n\hleq x_n\std{(a+2)}$, with also $z_0\hleq y_a(a+1)$.
Finally, $z_{i+1}\hleq z_i$ since each subrun $C_i\wstep{*}C_{i+1}$ is
single-pass on $\ttt$.

All this yields $x'_n\hleq s(\lambda_n)$ and
$v'_n\hleq 0^n$. Hence $S_{d,\lambda}$ is safe and robust: there is an
execution $(p,s(\alpha)\dol,0^{n}\dol,\dol)\wstep{*}
(r,s(\beta),0^{n'}\dol,\dol)$ if and only if $\alpha\in\LIM$, $s(\beta)\hleq
s(\alpha_n)$ and $n'\leq n$, entailing $H^\beta(n')\leq H^\alpha(n)$.
\\

There remains to consider $S^{-1}_{d,\lambda}$, the PCS component that
implements inverse Hardy steps for limits, see
\autoref{fig-pcs-invlim}. For given $a<d$, $S^{-1}_{d,\lambda,a}$
assumes that channel $\tto$ contains $s(\lambda_n)=y_d\ldots
y_{a+1}[y_a(a+1)]^n\std{(a+2)}$, guesses the position of the first
$y_a(a+1)$ factor, and checks that it indeed occurs $n$ times if
$\ttc$ contains $0^n$. This check uses copies $z_1, z_2, \ldots$ of
$y_a(a+1)$ temporarily stored on $\ttt$. Then $S^{-1}_{d,\lambda}$
writes back $s(\lambda)=y_d\ldots y_{a+1}z\std{a}$ on $\tto$, where
$z(a+1)=z_n$. The reader should be easily convinced that, as far as
one considers reliable steps, $S^{-1}_{d,\lambda}$ guarantees
\begin{equation}
\label{eq-s4-reliable}
(p,\: s(\alpha)\dol,\: 0^n\dol,\: \dol)\rstep{*} (r,u,v,w) \text{ iff }
\exists\lambda\in\LIM : \alpha=\lambda_n,
u=s(\lambda)\dol, v=0^{n}\dol, \text{ and } w=\dol
\:.
\end{equation}
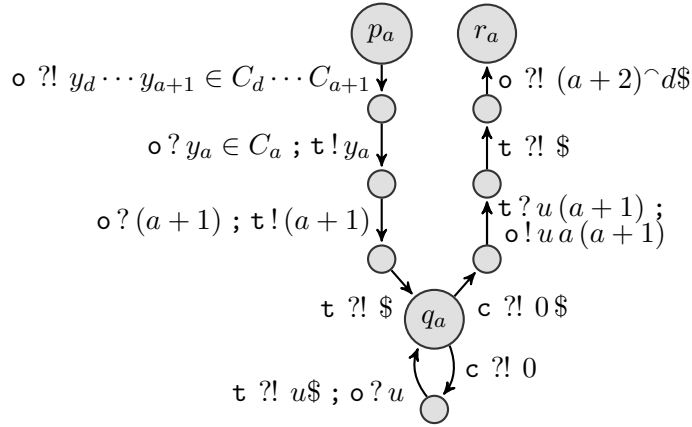
\begin{figure}[tbhp]
\centering
\scalebox{1.0}{
  \begin{tikzpicture}[->,>=stealth',shorten >=1pt,node distance=4cm,thick,auto,bend angle=30]
  \tikzstyle{every state}=[fill=white,draw=black,text=black]
\path
(0,0)   node [shape=circle,fill=black!12,draw=black!80] (p1) {$p_a$}
(0,-1) node [shape=circle,fill=black!12,draw=black!80] (p2) {}
(0,-2) node [shape=circle,fill=black!12,draw=black!80] (p3) {}
(0,-3) node [shape=circle,fill=black!12,draw=black!80] (p4) {}
(0.7,-3.8) node [shape=circle,fill=black!12,draw=black!80] (p5) {$q_a$}
(0.7,-5) node [shape=circle,fill=black!12,draw=black!80] (p5l) {}
(1.4,-3) node [shape=circle,fill=black!12,draw=black!80] (p6) {}
(1.4,-2) node [shape=circle,fill=black!12,draw=black!80] (p7) {}
(1.4,-1) node [shape=circle,fill=black!12,draw=black!80] (p8) {}
(1.4,-0) node [shape=circle,fill=black!12,draw=black!80] (p9) {$r_a$}
;

\path (p1) edge node [swap] {$\tto \,\qb\, y_d\cdots y_{a+1}\in C_d\cdots C_{a+1}$} (p2);
\path (p2) edge node [swap] {$\tto \,?\, y_a\in C_a \:\semicolon\: \ttt \,!\, y_a$} (p3);
\path (p3) edge node [swap] {$\tto \,?\, (a+1) \:\semicolon\: \ttt \,!\, (a+1)$} (p4);
\path (p4) edge node [swap,pos=0.55] {$\ttt \,\qb\, \dol$} (p5);
\path (p5) edge node [swap] {$\ttc \,\qb\, 0\,\dol$} (p6);
\path (p6) edge node [swap,pos=0.78] {$\ttt \,?\, u\,(a+1) \:\semicolon$} (p7);
\path (p6) edge node [swap,pos=0.22] {$\: \tto\,!\,u\,a\,(a+1)$} (p7);
\path (p7) edge node [swap] {$\ttt \,\qb\, \dol$} (p8);
\path (p8) edge node [swap] {$\tto \,\qb\, \std{(a+2)}\dol$} (p9);
\path (p5) edge [bend left] node [pos=0.40] {$\ttc \,\qb\, 0$} (p5l);
\path (p5l) edge [bend left] node [pos=0.572] {$\ttt \,\qb\, u\dol\:\semicolon\:\tto \,?\, u$} (p5);

  \end{tikzpicture}
}%
\caption{$S^{-1}_{d,\lambda,a}$, a PCS for inverse Hardy steps $(\lambda_n,n)\step{H^{-1}}(\lambda,n)$.}
\label{fig-pcs-invlim}
\end{figure}

When superseding is taken into account, a run from $p$ to $r$ in $S^{-1}_{d,\lambda}$
has the form $(p,s(\alpha)\dol,0^n\dol,\dol) \wstep{*} C_1 \wstep{*}
C_2 \wstep{*} \cdots C_n\wstep{*}(r,u,v,w)$ where, for $i=1,\ldots,n$,
$C_i$ is the $i$th configuration that visits state $q_a$.
Necessarily, $C_i$ is some $(q_a,x_i\dol x, 0^{n-i}0\dol v_i,
z_i\dol)$. The first visit to $q_a$ has $x\hleq y_d\ldots y_{a+1}$,
$z_1\hleq y_a(a+1)$ and $v_1=\epsilon$, the following ones ensure $x_i
= z_i x_{i+1}$, $z_{i+1}\hleq z_i$ and $v_{i+1}\hleq v_i 0$.
Concluding the run requires $x_n=\std{(a+2)}$.
Finally $v\hleq 0^n\dol$, $s(\beta)=y_d\ldots y_{a+1} (a+1)
z_1 \ldots z_{n-1}\std{(a+2)}$ and $u\hleq y_d\ldots y_{a+1}z \std{a}$
for $z(a+1)=z_n\hleq z_{n-1}\hleq\cdots z_2\hleq z_1\hleq y_a(a+1)$.
Thus $u=s(\beta)\dol$ and $v=0^{n'}\dol$ imply $s(\beta)\hleq
s(\lambda)$ for some $\lambda$ with $s(\lambda_n)\hleq s(\alpha)$,
yielding $H^\beta(n')\leq H^\lambda(n)=H^{\lambda_n}(n)\leq H^\alpha(n)$.

\subsection{Wrapping It Up}

\label{ssec-enc-wrapup}
With the above weak Hardy computers, we have the essential gadgets
required for our reductions.  The
wrapping-up %
is exactly as in~\cite{HSS-lics2012,phs-mfcs2010} (with a different
encoding and a different machine model) and will only be sketched.

\begin{theorem}[Verifying PCSs is Hard]\label{theo-PCS-e0hard}
  Reachability and Termination of PCSs are $\F_{\ez}$-hard.
\end{theorem}
\begin{proof}
We exhibit a \textsc{LogSpace} reduction from the halting problem of a
Turing machine $M$ working in $F_{\ez}$ space to the Reachability
problem in a PCS.  We assume wlog.\ $M$ to start in a state $p_0$ with
an empty tape and to have a single halting state $p_h$ that can only
be reached after clearing the tape.

\begin{figure}[tbp]
  \centering
  \begin{tikzpicture}
    [->,>=stealth',shorten >=1pt,node distance=.7cm,thick,auto,bend
      angle=30]
    \draw[rounded corners=4pt,dashed,fill=black!4,draw=black!20]
      (2.2,.8) rectangle (5.8,-.8);
    \node at (4,0)[font=\footnotesize,align=center,text width=1.6cm]
          {simulate $M$ with budget $B$};
    \path[every node/.style={shape=circle,fill=black!12,draw=black!80}]
    (-1,0) node (q0){$q_0$}
    (1,0) node[rounded corners=4pt,dashed,fill=black!4,draw=black!20,rectangle] (fw){$S_d$}
    (2.7,0) node (p0){$p_0$}
    (5.3,0) node (ph){$p_h$}
    (7,0) node[rounded corners=4pt,dashed,fill=black!4,draw=black!20,rectangle] (bw){$S^{-1}_d$}
    (9,0) node (qh){$q_h$};
    \path (q0) edge node[font=\footnotesize]{\begin{minipage}{4em}%
        \begin{tabular}{l}
          $\tto\ !\ \std{0}\dol$\\
          $\ttc\ !\ 0^n\dol$\\
          $\ttt\ !\ \dol$
        \end{tabular}
      \end{minipage}} (fw)
    (fw) edge (p0)
    (ph) edge (bw)
    (bw) edge node[font=\footnotesize]{\begin{minipage}{4em}%
        \begin{tabular}{l}
          $\tto\ ?\ \std{0}\dol$\\
          $\ttc\ ?\ 0^n\dol$\\
          $\ttt\ ?\ \dol$
        \end{tabular}
      \end{minipage}} (qh);
    \tikzset{every node}=[font=\footnotesize,text depth=.2ex,text height=2ex];
    \node at (1.5,-1.3){$\Omega_{d},n\hstep{H}\cdots\hstep{H}0,B$};
    \node at (6.5,-1.3){$0,B'\hstep{H^{\text{-}1}}\!\cdots\hstep{H^{\text{-}1}}\alpha,n'$};
  \end{tikzpicture}
  \caption{\label{fig-wrap}Schematics for \autoref{theo-PCS-e0hard}.}
\end{figure}
\autoref{fig-wrap} depicts schematically the PCS $S$ we construct for the
reduction.  Let $n\eqdef|M|$ and $d\eqdef n+1$.  A run in $S$ from the
initial configuration to the final one goes through three
stages:
\begin{enumerate}
\item The first stage robustly computes
  $F_{\ez}(|M|)=H^{\Omega_{d}}(n)$ by first writing
  $s(\Omega_{d})\dol$, \emph{i.e.} $\std{0}\dol$, on $\tto$, $0^n\dol$ on $\ttc$, and
  $\dol$ on $\ttt$, then by using $S_d$ to
  perform forward Hardy steps; thus upon reaching state $p_0$, $\tto$
  and $\ttt$ contain $\dol$ and $\ttc$ encodes a budget $B\leq
  F_{\ez}(|M|)$.
\item\label{stage-2} The central component simulates
  $M$ over $\ttc$ where the symbols $0$ act as blanks---this is
  easily done by cycling through the channel contents to simulate the
  moves of the head of $M$ on its tape.  Due to superseding steps, the
  outcome %
  upon reaching $p_h$ is that $\ttc$ contains
  $B'\leq B$ symbols $0$.
\item The last stage robustly computes $(F_{\ez})^{-1}(B')$ by running
  $S^{-1}_d$ to perform backward Hardy steps.  This leads to
  $\tto$ containing the encoding of some ordinal $\alpha$ and $\ttc$
  of some $n'$, but we empty these channels and check that
  $\alpha=\Omega_d$ and $n'=n$ before entering state $q_h$.
\end{enumerate}
Because
  $H^{\Omega_d}(n)\geq B\geq B'\geq H^\alpha(n')=H^{\Omega_d}(n)$, %
all the inequalities are actually equalities, and the simulation of
$M$ in stage~\ref{stage-2} has necessarily employed reliable steps.
Hence, $M$ halts if and only if $(q_h,\epsilon,\epsilon,\epsilon)$
is reachable from $(q_0,\epsilon,\epsilon,\epsilon)$ in $S$.

The case of (non-)Termination is similar, but employs a \emph{time}
budget in a separate channel in addition to the space budget, in order
to make sure that the simulation of $M$ terminates in all cases, and
leads to a state $q_h$ that is the only one from which an infinite run
can start in $S$.
\end{proof}

\section{Alternative Semantics}\label{sec-related}
In this section we consider variant models for channel systems with
priorities or losses and compare them with our PLCS model. The aim is
to better understand the consequences, or lack thereof, of our
choices.

We first consider \emph{strict-superseding} systems, where messages
may only supersede messages of strictly lower priority, and
\emph{overtaking} systems, where higher priority messages may move
ahead of lower priority messages instead of erasing them. For
completeness, we also discuss systems based on \emph{priority queues},
where overtaking of lower priority messages is mandatory.

\subsection{Strict-Superseding and Overtaking}\label{ssec-stricter-semantics}
In this section, we discuss two alternative operational semantics for
PCSs that may seem more natural than our standard $\TS_\hh$.

\begin{description}
\item[Strict-Superseding Semantics]
Here, a high-priority message may only supersede messages
having \emph{strictly} lower priority. Formally, we replace the
internal-superseding relation $C\hstep{\hh k}C'$ with a new
superseding relation, denoted $C\Hstep{\hh k}C'$, and based on
\[
x\Hstep{\hh k}y
\;\equivdef\;
x=a_1\cdots a_\ell
\,\land\, y=a_1\cdots a_{k-1}\cdot a_{k+1}\cdots a_\ell
\,\land\, a_k<a_{k+1}
\;.
\]
Equivalently, we replace the rewrite rules from
Eq.~\eqref{eq-rewr-4-hstep} with $\bigl\{ a\,a'\to a' ~\big|~ 0\leq a < a'\leq d \bigr\}$.\medskip

\item[Overtaking semantics]
Here, a high-priority message may move ahead of a low-priority message
but this does not erase the low-priority message.
 Formally, we replace  $C\hstep{\hh k}C'$ with  $C\Ostep{\hh k}C'$,  based on
\[
x\Ostep{\hh k} y
\;\equivdef\;
x=a_1\cdots a_k\,a_{k+1}\cdots a_\ell
\,\land\, y=a_1\cdots a_{k+1}\,a_k\cdots a_\ell
\,\land\, a_k<a_{k+1}
\;.
\]
In rewriting terms, $\Ostep{}$ is defined by the rules $\bigl\{ a a'\to
a'\,a ~ \big| ~ 0\leq a< a'\leq d \bigr\}$.
\end{description}

\noindent These two mechanisms drop fewer messages than our internal-superseding
semantics. They may however be inadequate in case of network
congestion, for instance they offer no solutions if all the messages
in the congested buffers have the same priority. In any case, we show
below that verification is undecidable for these two variant semantics
(one can simulate Turing-powerful reliable channel systems with PCS
under strict-superseding or overtaking semantics), which explains our
choice of internal-superseding semantics.

\begin{theorem}\label{thm-different-semantics}
Reachability and Termination are undecidable for PCSs under both the
strict-superseding and overtaking semantics.
\end{theorem}
\begin{proof}
We reduce from the reachability problem for channel systems with
reliable channels which is known to be undecidable~\cite{brand83}.
A system $S$ with reliable channels
uses a finite (un-prioritized) alphabet
$\Sigma=\{a_0,\ldots,a_{p-1}\}$ and is equipped with $m$ channels
$\ttc_1, \ldots, \ttc_m$. We simulate $S$ with a PCS $S'$ with strict
superseding semantics having the same $m$ channels and using the
$\Sigma_d$ priority alphabet with $d=p$. We use $d\in\Sigma_d$ as a
separator, denoted $\dol$ for clarity, while the other priorities
$i\in \{0,\ldots,p-1\}$ represent the original messages $a_i$. A
string $w=a_{i_1}\cdots a_{i_n}\in \Sigma^*$ is encoded as
$\widetilde{w}\eqdef i_1\,\dol\cdots i_n\,\dol\in\Sigma_d^*$. The
actual reduction is obtained by equipping $S'$ with transition rules
that simulate the rules of $S$ as illustrated in \autoref{fig-cs2pcs}.
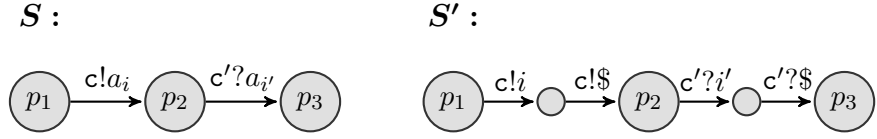
\begin{figure}[tbp]
\centering
\scalebox{1.0}{
  \begin{tikzpicture}[->,>=stealth',shorten >=1pt,thick,auto,bend angle=30]
\path
(0,3) node [shape=circle,fill=black!12,draw=black!80] (p1) {$p_1$}
(1.8,3) node [shape=circle,fill=black!12,draw=black!80] (p2) {$p_2$}
(3.6,3) node [shape=circle,fill=black!12,draw=black!80] (p3) {$p_3$} 
node [above of=p1,node distance=3em] {\large $\bm{S:}$}
;
\path (p1) edge node {$\ttc  ! a_i$} (p2) ;
\path (p2) edge node {$\ttc' ? a_{i'}$} (p3) ;

\path
(5.5,3) node [shape=circle,fill=black!12,draw=black!80] (p1) {$p_1$}
(6.8,3) node [shape=circle,fill=black!12,draw=black!80] (p11) {}
(8.1,3) node [shape=circle,fill=black!12,draw=black!80] (p2) {$p_2$}
(9.4,3) node [shape=circle,fill=black!12,draw=black!80] (p21) {}
(10.7,3) node [shape=circle,fill=black!12,draw=black!80] (p3) {$p_3$}
node [above of=p1,node distance=3em] {\large $\bm{S':}$}
;

\path (p1)  edge node {$\ttc  ! i$} (p11) ;
\path (p11) edge node {$\ttc  ! \dol$} (p2) ;
\path (p2)  edge node {$\ttc' ? i'$} (p21) ; 
\path (p21) edge node {$\ttc' ? \dol$} (p3) ;

  \end{tikzpicture}
}%
\caption{Simulating reliable channels (left) with ``strict-superseding'' or ``overtaking'' PCSs (right).}
\label{fig-cs2pcs}
\end{figure}
In essence, where $S$ would write $a_i$, $S'$ will write $i$ followed
by $\dol$, \textit{i.e.}, $\widetilde{a_i}$, and $S'$ will read $i'\cdot \dol$
where $S$ would read $a_{i'}$.

With the \emph{strict superseding} policy, the only superseding that
can occur is to have $\dol$ erase a preceding $i<\dol$.  This results
in a channel containing two or possibly more consecutive $\dol$
symbols, a pattern that will never vanish in this simulation and that
eventually forbids reading on the involved channel.  In particular,
any run of $S'$ that reaches a final configuration
$C_\final=(q_\final,\epsilon,\ldots,\epsilon)$ has not used any strict
superseding and thus corresponds to a run of the reliable channel
system $S$. The same reduction works for Termination.

With the overtaking semantics, only $\dol$ can overtake ``original''
messages of $S$ in $\widetilde{w}$. However, such an overtake results
in having two consecutive $\dol$ symbols on the channel, a pattern
that can never disappear. Behind the $\dol\dol$ block, two lower
messages $0\leq i,j<d$ may occur consecutively and be open to
overtaking but this will not derail the simulation since $S'$ cannot
read beyond $\dol\dol$.
\end{proof}

\subsection{Channels as Priority Queues}\label{ssec-priority-queues}

For the sake of completeness, let us mention channel systems where
channels behave as priority queues. Here, reading from a channel will
always read a message having highest priority among the contents of
the channel. This can be seen as an extreme version of the overtaking
semantics, where overtaking is mandatory.
Such a model is not relevant for our purposes since it is not meant to
handle congested communication links: instead, it uses priorities as a
way of choosing in which order messages should be processed.

Let us mention that finite-state communicating systems with priority
queues can easily simulate Minsky machines by using a queue for each
counter and two priorities: high-priority messages encode the counter
value in unary, while a low-priority message can only be read in case
of a zero-valued counter.  They are hence Turing-powerful.

The most relevant case however is that of a single communication bus
where many processes can read and write.  Because only one queue is
available---and still assuming a singleton alphabet for message
contents---it is easy to see that priority queues systems are
equivalent to Minsky counter machines restricted to \emph{nested
zero-tests}. Recall that for a machine with $m$ counters
$\ttc_1,\ldots,\ttc_m$, nested zero-tests are tests of the form
``$(\ttc_1=0 \land
\ttc_2=0\land\cdots\land\ttc_i=0)?$'', \textit{i.e.}, one can only
test the $i$th counter for emptiness when already the previous
counters are empty. While Minsky counter machines with arbitrary
zero-tests are Turing-powerful, they are equivalent to Petri nets when
zero-tests are forbidden. Minsky machines with nested zero-tests are
an intermediary model for which reachability is known to be decidable,
see~\cite{reinhardt2008} and \citep[Chapter~5]{bonnet-phd2013}.

\section{Higher-Order Lossy Channel Systems}\label{ssec-holcs}
{\renewcommand{\enc}[1]{\lfloor #1\rfloor}
In this section, we introduce higher-order lossy channel systems,
aka HOLCSs, a family of models that extend lossy channel
systems. While a higher-order pushdown
automaton has a stack of stacks of \ldots\  of stacks~\cite{aho69}, a
HOLCS has lossy channels inside lossy channels inside \ldots\  inside
lossy channels.
(In this setting, the ``dynamic'' lossy channel systems
from~\cite{AAC-fsttcs2012} are a special case of 2LCSs, or second-order LCSs.)
HOLCSs are well-structured,
see \autoref{thm-holcs-wsts}, hence enjoy the usual decidability
results from well-structured systems theory.

Our main result is that PCSs can simulate HOLCSs, see
\autoref{thm-pcs-simul-holcs}. On the one hand, this underlines the
expressive power and naturalness of PCSs, in particular since the
reductions we provide are quite straight-forward. On the other hand,
we immediately obtain undecidability results: problems like boundedness or
repeated control-state reachability are undecidable for PCSs since
they 
already are for (first-order) LCSs
(see~\cite{mayr03,phs-rp2010}). 

%

\subsection{Syntax}
Formally, and given $k\in\Nat$, a \emph{$k$th-order LCS} $S=(k,\Sigma,\Ch,Q,\Delta)$ has
first-order, second-order, \ldots, up to $k$th-order, channels. We
assume for simplicity that $S$ has, for all $n=1,\ldots,k$, the same
number $m$ of $n$th-order channels, denoted
$\ttc_{n,1},\ldots,\ttc_{n,m}$. Standardly, $S$ uses a finite
(un-prioritized) alphabet $\Sigma=\{a_1,\ldots,a_p\}$. Write
$\Sigma^{*1}$ for the set of finite sequences of messages (usually
written just $\Sigma^*$), and $\Sigma^{*(n+1)}$ for the set of finite
sequences of elements from $\Sigma^{*n}$. We further order each
$\Sigma^{*(n+1)}$ with $\leq_{*(n+1)}$, \textit{i.e.}, the sequence extension
of $(\Sigma^{*n},\leq_{*n})$, equating $(\Sigma^{*0},\leq_{*0})$ with
$(\Sigma,=)$. Precisely, given two sequences $x=x_1\ldots x_\ell$ and
$y=y_1\ldots y_m$ in $\Sigma^{*(n+1)}$, we let
\[
x_1\ldots x_\ell\leq_{*(n+1)}y_1\ldots y_m
\equivdef
\exists 1\leq i_1<i_2< \cdots < i_\ell\leq m:
x_1\leq_{*n} y_{i_1}\land \cdots\land x_\ell\leq_{*n} y_{i_\ell}
\]
Using Higman's Lemma and induction over $n$, one sees that
$(\Sigma^{*n},\leq_{*n})$ is a well-quasi-order for any $n$.

At any given time, the contents of a $n$th-order
channel is a sequence $w\in\Sigma^{*n}$, so that a configuration of
$S$ has the form
$C=(q,x_{1,1},\ldots,x_{1,m},\ldots,x_{k,1},\ldots,x_{k,m})$ with $q$
a control state and $x_{n,i}\in\Sigma^{*n}$ for all $1\leq n\leq k$
and $1\leq i\leq m$.
These configurations are ordered by
\[
(q,x_{1,1},\ldots,x_{k,m})
\holeq
(q',y_{1,1},\ldots,y_{k,m})
\equivdef
q=q'
\;\land\;
\forall n,i: x_{n,i}\leq_{*n}y_{n,i}
\:.
\]
This ordering of configurations is a wqo since, for each $n=1,\ldots,k$,
\mbox{$(\Sigma^{*n},\leq_{*n})$} is.

\subsection{Semantics}
A HOLCS $S$ as above has a \emph{reliable} transition
$C\step{\delta}C'$ between configurations
$C=(q,\ldots,x_{i,n}\ldots)$ and $C'=(q',\ldots,y_{i,n},\ldots)$ if
$\delta=(q,\op,q')$ is a rule moving control from $q$ to $q'$ and if the channel contents
are modified according to the operation carried by $\delta$, as we now define. There
are four cases:
\begin{description}
\item[$\bm{\op=\ttc_{1,i}! a}$\nocolon] for some $a\in \Sigma$ and $1\leq i\leq m$:
then $y_{1,i}=x_{1,i}\cdot a$ while $y_{j,n}=x_{j,n}$ when $n>1$ or
$j\not=i$, \textit{i.e.}, one writes a message to a 1st-order channel as in
standard channel systems;
\item[$\bm{\op=\ttc_{1,i}? a}$] then one reads a message from a 1st-order
channel, \textit{i.e.}, $x_{1,i}=a\cdot y_{1,i}$ while the other channels are
untouched;
\item[$\bm{\op=\hopush{n+1,i}{n,j}}$\nocolon] for some $1\leq n < k$: then one
appends a copy of the whole contents of $\ttc_{n,j}$ (a $n$th-order
channel) to $\ttc_{n+1,i}$, where it becomes the last element of the
higher-order sequence. Formally, $y_{n+1,i}= x_{n+1,i}\cdot x_{n,j}$ and
the other channels are untouched;
\item[$\bm{\op=\hopop{n,i}{n+1,j}}$] then one moves the first element of the
higher-order sequence currently in $\ttc_{n+1,j}$ to $\ttc_{n,i}$
where it becomes the whole contents (the previous contents is erased).
Formally, if $u\in\Sigma^{*n}$ is the first element of $x_{n+1,j}$,
then $y_{n,i}=u$, $u\cdot y_{n+1,j}=x_{n+1,j}$, and the other channels
are untouched.
\end{description}
In addition, all steps $C\step{}C'$ for $C'\holt C$,
called \emph{losing steps}, are allowed. This states that at any time
the system may lose individual messages, sequences of messages,
sequences of sequences of \ldots\ of messages, anywhere inside the
channels.

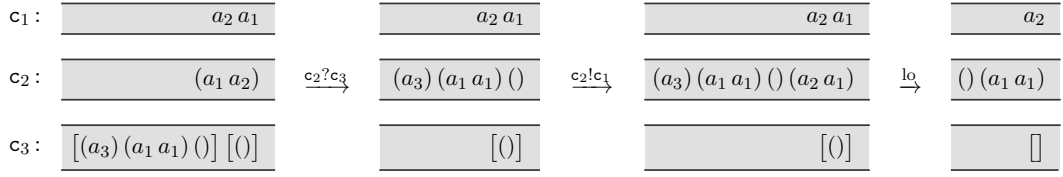
\begin{figure}[tbp]
\centering
\scalebox{.8}{\begin{tikzpicture}[->,>=stealth',shorten >=1pt,node
      distance=1cm,thick,auto,bend angle=30]
\node[fill=black!12,text width=8.5em,align=flush right]
    (c0){$a_2\,a_1\;$};
\node[below=1.1em of c0,fill=black!12,text width=8.5em,align=flush right]
    (c1){$(a_1\,a_2)\;$};
\node[below=1.1em of c1,fill=black!12,text width=8.5em,align=flush right]
    (c2){$\bigl[(a_3)\,(a_1\,a_1)\,()\bigr]\, \bigl[()\bigr]\;$};
\node[right=4.5em of c0,fill=black!12,text width=6.2em,align=flush right]
    (d0){$a_2\,a_1\;$};
\node[below=1.1em of d0,fill=black!12,text width=6.2em,align=flush right]
    (d1){$(a_3)\,(a_1\,a_1)\,()\;$};
\node[below=1.1em of d1,fill=black!12,text width=6.2em,align=flush right]
    (d2){$\bigl[()\bigr]\;$};
\node[right=4.5em of d0,fill=black!12,text width=9em,align=flush right]
    (e0){$a_2\,a_1\;$};
\node[below=1.1em of e0,fill=black!12,text width=9em,align=flush right]
    (e1){$(a_3)\,(a_1\,a_1)\,()\,(a_2\,a_1)\;$};
\node[below=1.1em of e1,fill=black!12,text width=9em,align=flush right]
    (e2){$\bigl[()\bigr]\;$};
\node[right=3.5em of e0,fill=black!12,text width=4.0em,align=flush right]
    (f0){$a_2\;$};
\node[below=1.1em of f0,fill=black!12,text width=4.0em,align=flush right]
    (f1){$()\,(a_1\,a_1)\;$};
\node[below=1.1em of f1,fill=black!12,text width=4.0em,align=flush right]
    (f2){$\bigl[\bigr]\;$};
\node[left=.5em of c0]{$\ttc_1\,$:};
\node[left=.5em of c1]{$\ttc_2\,$:};
\node[left=.5em of c2]{$\ttc_3\,$:};
\draw[-,thick,color=black!80]
  (c0.north west) -- (c0.north east)	  (c0.south west) -- (c0.south east)
  (c1.north west) -- (c1.north east)	  (c1.south west) -- (c1.south east)
  (c2.north west) -- (c2.north east)	  (c2.south west) -- (c2.south east)
  (d0.north west) -- (d0.north east)	  (d0.south west) -- (d0.south east)
  (d1.north west) -- (d1.north east)	  (d1.south west) -- (d1.south east)
  (d2.north west) -- (d2.north east)	  (d2.south west) -- (d2.south east)
  (e0.north west) -- (e0.north east)	  (e0.south west) -- (e0.south east)
  (e1.north west) -- (e1.north east)	  (e1.south west) -- (e1.south east)
  (e2.north west) -- (e2.north east)	  (e2.south west) -- (e2.south east)
  (f0.north west) -- (f0.north east)	  (f0.south west) -- (f0.south east)
  (f1.north west) -- (f1.north east)	  (f1.south west) -- (f1.south east)
  (f2.north west) -- (f2.north east)	  (f2.south west) -- (f2.south east)
;
\node at ($(c1.east)!0.5!(d1.west)$) {$\step{\hopopii}$};
\node at ($(d1.east)!0.5!(e1.west)$) {$\step{\hopushi}$};
\node at ($(e1.east)!0.5!(f1.west)$) {$\step{\loss}$};
  \end{tikzpicture}
}%
\caption{Reading, writing, and losing in HOLCSs.}
\label{fig-holcs-steps}
\end{figure}
\autoref{fig-holcs-steps} illustrates the behavior of higher-order
channels under reads, writes and losses (control states omitted).

\begin{theorem}
\label{thm-holcs-wsts}
HOLCSs equipped with $\holeq$ are well-structured.
\end{theorem}
\begin{proof}
The ordering $C\holeq C'$ entails $C'\step{*}C$ (via losing steps).
\end{proof}

\subsection{Simulation by PCSs}

Let us consider a $k$th order LCS $S$ with $\Sigma=\{a_1,\ldots,a_p\}$ and
$k\cdot m$ channels. We assume that $m=1$ in order to simplify our
constructions and proofs but they extend directly to the more
expressive cases where $m>1$.

We simulate $S$ with a $d$-PCS $\wS$ having $k$ channels and
$d\eqdef p+k-1$.
In $\Sigma_d$, the lower priorities $0,\ldots,p-1$
will be denoted $a_1,\ldots,a_p$ since they directly encode the
messages from $\Sigma$, while the higher priorities $p,\ldots,p+k-1$
will be denoted $\dol_1,\ldots,\dol_k$ since they are used as
separators.

\subsubsection{Encoding Configurations}
A configuration $C=(q,x_1,\ldots,x_k)$ of $S$ is encoded as
$\wC\eqdef(q,\enc{x_1}_1,\ldots,\enc{x_k}_k)$, where,
for $n=1,\ldots,k$, $\enc{x}_n$ denotes the \emph{$n$th-level
encoding} of a $n$th-level sequence $x=u_1 u_2\cdots
u_\ell\in\Sigma^{*n}$. Encodings are defined with
\begin{align*}
\enc{a_{i_1}a_{i_2}\cdots a_{i_\ell}}_1&\eqdef
\dol_1 a_{i_1}\dol_1 a_{i_2}\dol_1 \cdots
a_{i_\ell}\dol_1              \:,
\\
\enc{u_1u_2\cdots u_\ell}_{n+1}&\eqdef
\dol_{n+1}\enc{u_1}_n\dol_{n+1}\enc{u_2}_n\dol_{n+1}\cdots
\enc{u_\ell}_n\dol_{n+1}   \:.
\end{align*}
Note that $\enc{\epsilon}_n=\dol_n$ and differs from
$\epsilon\in\Sigma_d^*$. For $n=1,\ldots,k$, we let
$E_n\eqdef\{\enc{x}_n ~|~ x\in\Sigma^{*n}\}$ denote the set of all
$n$-level encodings. These are regular languages that are
captured by the following regular expressions:
\begin{xalignat*}{2}
E_1&\eqdef \dol_1\cdot\bigl((a_1+\cdots+a_p)\cdot \dol_1\bigr)^*
\:,
&
E_{n+1}&\eqdef \dol_{n+1}\cdot\bigl(E_n \cdot \dol_{n+1}\bigr)^*
\:.
\end{xalignat*}
With this encoding, $(\Sigma^{*n},\leq_{*n})$ and $(E_n,\hleq)$ are
isomorphic:
\begin{lemma}
\label{lem-encoding-is-iso}
For any $x,y\in\Sigma^{*n}$, $x\leq_{*n}y$ if, and only if,
$\enc{x}_n\hleq\enc{y}_n$.
\end{lemma}
\begin{proof}[Proof Idea]
By induction on $n$. For the ``$\Leftarrow$''
direction it is easier to rely on priority embedding, \textit{i.e.},
prove that $\enc{x}_n\pleq\enc{y}_n$ implies $x\leq_{*n}y$ and
apply \autoref{lem-ghleq-gpleq-equivalence}.
\end{proof}

\subsubsection{Encoding First-Order Rules}
\begin{figure}[tbp]
\centering
\scalebox{.8}{\begin{tikzpicture}[->,>=stealth',shorten >=1pt,node distance=1cm,thick,auto,bend angle=30]
\node[stt](q1){$q_1$};
\node[stt,right=4em of q1](q2){$q_2$};
\path (q1) edge node {${\ttc_{1}\,!\,a_i}$} (q2) ;
\node[stt,below=3em of q1](q3){$q_3$};
\node[stt,right=4em of q3](q4){$q_4$};
\path (q3) edge node {${\ttc_{1}\,?\,a_j}$} (q4) ;
\node[stt,right=9em of q2](Q1){$q_1$};
\node[stt,below=3em of Q1](Q3){$q_3$};
\node[stt,right=4em of Q1](Q1x) {};
\node[stt,right=4em of Q1x](Q2) {$q_2$};
\node[stt,right=4em of Q3](Q3x) {};
\node[stt,right=4em of Q3x](Q4) {$q_4$};
\path (Q1) edge node {$\ttc_1 \,!\, a_i$} (Q1x) ;
\path (Q1x) edge node {$\ttc_1 \,!\, \dol_1$} (Q2) ;
\path (Q3) edge node {$\ttc_1 \,?\, \dol_1$} (Q3x) ;
\path (Q3x) edge node {$\ttc_1 \,?\, a_j$} (Q4) ;
\node[left=0.5em of q1] {\large $\delta_1$:};
\node[left=0.5em of q3] {\large $\delta_2$:};
\path (q1) -- (q2) node [midway] (midS) {};
\path (Q1) -- (Q2) node [midway] (midSp) {};
\node[above=3em of midS,fill=black!12,text width=5em,text height=0.9em,align=flush right] (cn1){$a_3\,a_1\,a_4\;$};
\node[above=3em of midSp,fill=black!12,text width=9em,text height=0.9em,align=flush right] (dn1){$\dol_1\,a_3\,\dol_1\, a_1\,\dol_1\, a_4\,\dol_1\;$};
\node[left=0.5em of cn1]{$\ttc_1\,$:};
\node[left=0.5em of dn1]{$\ttc_1\,$:};
\node[below=-0.4em of cn1] {$\vdots$};
\node[below=-0.4em of dn1] {$\vdots$};
\draw[-,thick,color=black!80]
    (cn1.north west) -- (cn1.north east)
    (cn1.south west) -- (cn1.south east)
    (dn1.north west) -- (dn1.north east)
    (dn1.south west) -- (dn1.south east);
\node[above left=1em and 2em of q1] {\large $\bm{S:}$};
\node[above left=1em and 2em of Q1] {\large $\bm{\wS:}$};
  \end{tikzpicture}
}%
\caption{Simulation of HOLCS $S$ with PCS $\wS$: first-order rules.}
\label{fig-simul1-holcs}
\end{figure}
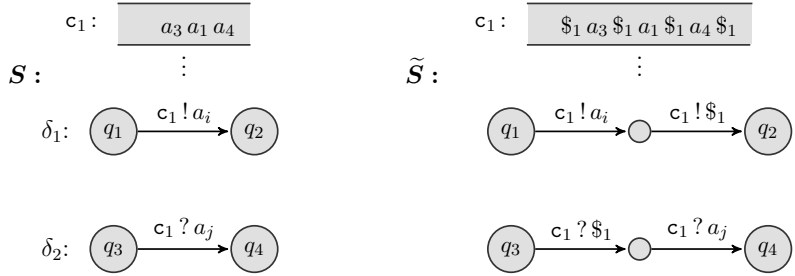
We may now complete the definition of $\wS$ by describing its rules,
with the goal of simulating the operational semantics of $S$ while
working on $\wC$ and on encodings of sequences from some
$\Sigma^{*n}$. This is easy for 1st-order rules that operate on
$\ttc_1$ only: where $S$ writes $a_i$, $\wS$ writes $a_i\cdot\dol_1$.
Where it reads $a_j$, $\wS$ reads $\dol_1 \cdot a_j$. This is
illustrated in \autoref{fig-simul1-holcs}.

\subsubsection{Encoding Higher-Order Rules}
Higher-order rules of the form $q_1\step{\hopushn}q_2$ are
simulated in $\wS$ as we illustrate in case $\delta_3$
of \autoref{fig-simul2-holcs}. Here $\wS$ uses a loop, abbreviated as
$\times\step{\ttc_n?u\semicolon\ttc_n!u\semicolon\ttc_{n+1}!u}\times$,
to append a copy of $\ttc_n$'s contents to $\ttc_{n+1}$. This uses a
high-priority $\dol_{n+1}$ to mark the end of $\ttc_n$'s contents and
ensure that all of it has been read (and written back). Note that the
loop checks that $u$ is in $E_n$, \textit{i.e.}, is a well-formed
encoding, which is done by following a DFA for $E_n$. When the
transfer is completed, a $\dol_{n+1}$ must be appended to $\ttc_{n+1}$
to ensure consistency.
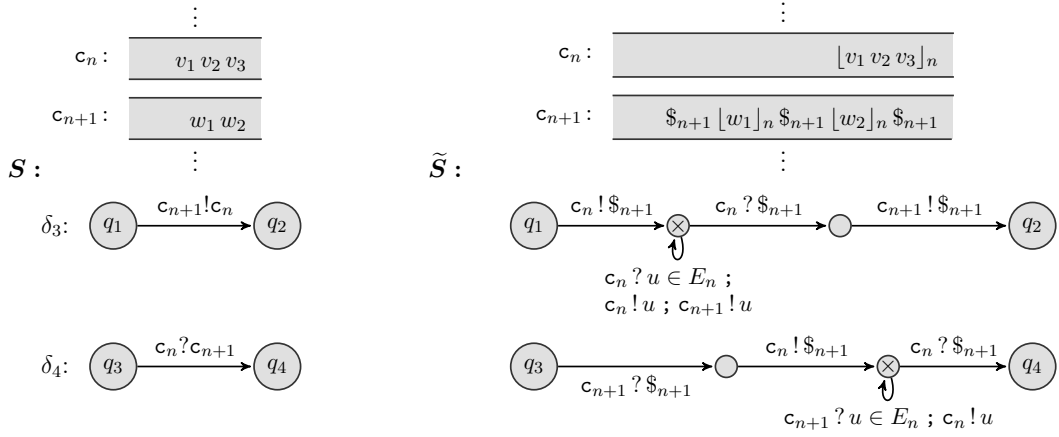
\begin{figure}[htbp]
\centering
\scalebox{.8}{\begin{tikzpicture}[->,>=stealth',shorten >=1pt,node distance=1cm,thick,auto,bend angle=30]
\node[stt](q1){$q_1$};
\node[stt,right=5em of q1](q2){$q_2$};
\path (q1) edge node {$\hopushn$} (q2) ;
\node[stt,below=4em of q1](q3){$q_3$};
\node[stt,right=5em of q3](q4){$q_4$};
\path (q3) edge node {$\hopopn$} (q4) ;
\node[stt,right=9em of q2](Q1){$q_1$};
\node[stt,right=1.8cm of Q1](Qx){};
\node[stt,right=2.3cm of Qx](Qy){};
\node[stt,right=2.6cm of Qy](Q2){$q_2$};
\path (Q1)  edge node {$\ttc_n \,!\, \dol_{n+1}$} (Qx) ;
\path (Qx)  edge node {$\ttc_n \,?\, \dol_{n+1}$} (Qy) ;
\path (Qy)  edge node {$\ttc_{n+1} \,!\, \dol_{n+1}$} (Q2) ;
\path (Qx)  edge [loop below] node [align=left] {$\ttc_n \,?\, u\in E_n\:\semicolon$\\$\ttc_n \,!\, u\:\semicolon\:\ttc_{n+1}\,!\, u$} (Qx) ;
\node[stt,right=9em of q4](Q3){$q_3$};
\node[stt,right=2.6cm of Q3](Qz){};
\node[stt,right=2.3cm of Qz](Qt){};
\node[stt,right=1.8cm of Qt](Q4){$q_4$};
\path (Q3)  edge node [swap] {$\ttc_{n+1} \,?\, \dol_{n+1}$} (Qz) ;
\path (Qz)  edge node {$\ttc_n \,!\, \dol_{n+1}$} (Qt) ;
\path (Qt)  edge node {$\ttc_n \,?\, \dol_{n+1}$} (Q4) ;
\path (Qt)  edge [loop below] node [align=left] {$\ttc_{n+1} \,?\, u\in E_n\:\semicolon\:\ttc_n \,!\, u$} (Qt) ;
\node[left=0.5em of q1] {\large $\delta_3$:};
\node[left=0.5em of q3] {\large $\delta_4$:};
\node at (Qx) {$\times$};
\node at (Qt) {$\times$};
\path (q1) -- (q2) node [midway] (midS) {};
\path (Q1) -- (Q2) node [midway] (midSp) {};
\node[above=3em of midS,fill=black!12,text width=5em,text height=0.9em,align=flush right] (cn1){$w_1\,w_2\;$};
\node[above=0.8em of cn1,fill=black!12,text width=5em,text height=0.9em,align=flush right] (cn){$v_1\,v_2\,v_3\;$};
\node[above=3em of midSp,fill=black!12,text width=14em,text height=0.9em,align=flush right] (dn1){$\dol_{n+1}\,\enc{w_1}_n\,\dol_{n+1}\, \enc{w_2}_n\,\dol_{n+1}\;$};
\node[above=0.8em of dn1,fill=black!12,text width=14em,text height=0.9em,align=flush right] (dn){$\enc{v_1\, v_2\, v_3}_n\;$};
\node[left=0.5em of cn1]{$\ttc_{n+1}\,$:};
\node[left=0.5em of cn]{$\ttc_n\,$:};
\node[left=0.5em of dn1]{$\ttc_{n+1}\,$:};
\node[left=0.5em of dn]{$\ttc_n\,$:};
\node[above=0.1em of cn] {$\vdots$};
\node[below=-0.4em of cn1] {$\vdots$};
\node[above=0.1em of dn] {$\vdots$};
\node[below=-0.4em of dn1] {$\vdots$};
\draw[-,thick,color=black!80]
    (cn.north west) -- (cn.north east)
    (cn.south west) -- (cn.south east)
    (cn1.north west) -- (cn1.north east)
    (cn1.south west) -- (cn1.south east)
    (dn.north west) -- (dn.north east)
    (dn.south west) -- (dn.south east)
    (dn1.north west) -- (dn1.north east)
    (dn1.south west) -- (dn1.south east);
\node[above left=1em and 2em of q1] {\large $\bm{S:}$};
\node[above left=1em and 2em of Q1] {\large $\bm{\wS:}$};
  \end{tikzpicture}
}%
\caption{Simulation of HOLCS $S$ with PCS $\wS$: higher-order rules.}
\label{fig-simul2-holcs}
\end{figure}

Simulating rules of the form $q_3\step{\hopopn}q_4$ follows
the same logic (see \autoref{fig-simul2-holcs}): $\wS$ reads the first
$n$-level encoding in $\ttc_{n+1}$, checking that is is well-formed
(with ``$u\in E_n$'') and writes it to $\ttc_n$. Simultaneously, the
previous contents of $\ttc_n$ is emptied by writing a $\dol_{n+1}$ and
reading it.

\subsubsection{Correctness}
The correctness of this simulation is captured by the next two propositions.
\begin{proposition}
\label{prop-pcs-simul-holcs}
If $S$ has a run $C\step{*}C'$, then $\wS$ has a run
 $\wC\step{*}\wCp$.
\end{proposition}
\begin{proof}[Proof Idea]
On the one hand, $\wS$ has been designed so that its
behavior \emph{without any superseding} directly mimics on encodings
the effect of the reliable steps $C\step{} C'$. Then lossy steps in
$S$ are simulated by superseding since, thanks
to \autoref{lem-encoding-is-iso}, $C\hogeq C'$ entails
$\wC\hgeq\wCp$.
\end{proof}
There is an exact reciprocal to \autoref{prop-pcs-simul-holcs}:
\begin{proposition}
\label{prop-holcs-simul-pcs}
If $\wS$ has a run $\wC\step{*}\wCp$ then $S$ has a
run $C\step{*}C'$.
\end{proposition}
The correctness proof is  harder in this direction since the steps of
$\wS$ are finer-grained than the steps of $S$. Note that a
configuration $D=(q,w_1,\ldots,w_k)$ of $\wS$ is not
necessarily the encoding $\wC$ of a configuration of $S$: if
$q$ is not an original state of $S$ (\textit{i.e.}, is one of the
unnamed states depicted on the right-hand side
of \autoref{fig-simul1-holcs} or \autoref{fig-simul2-holcs}) then $D$
is not a $\wC$.
Furthermore, if $q$ is an original state, it is possible that some
$w_n$ does not belong to $E_n$.

With these difficulties in mind, we say that $D\in\Conf_\wS$ is
\emph{safe} if every $w_n$ ends with a $\dol_n$ (and contains no
$\dol_{n'}$ for $n'>n$), and that $D$ is a \emph{stable} configuration
if $q$ is an original state. We rely on the write-superseding
semantics (see \autoref{ssec-write-superseding}) for a better control
on the form of the runs. We say that a write-superseding run
$D_0\wstep{}D_1\wstep{}\cdots \wstep{} D_r$ in $\wS$ is a
\emph{macro-step} if $D_0$ and $D_r$ are the only stable
configurations it visits (in essence, a macro-step just follows the
rules introduced in $\wS$ to simulate a single rule of $S$ using
write-superseding semantics). We are now ready for the following
lemmata.
\begin{lemma}
\label{lem-safe}
Let $D\wstep{+} D'$ be any macro-step in $\wS$. If $D$ is safe
then $D'$ is safe too.
\end{lemma}
\begin{proof}
Since the last message in a channel cannot be superseded, a $\ttc_n$
whose contents is safe remains safe if one does not read
its final $\dol_n$, or one appends some safe contents.
We now consider all four cases for the macro-step $D\wstep{+}D'$.
\begin{enumerate}[label=$\delta_{\arabic*}$:]
\item
using rules of the form $q_1\step{\ttc_1!a_i}\step{\ttc_1!\dol_1}q_2$,
it writes a safe $a_i\cdot \dol_1$ in $\ttc_1$ (and does not
read from the other channels).

\item
reading with some $q_3\step{\ttc_1?\dol_1}\step{\ttc_1?a_j}q_4$, the
read $\dol_1$ cannot be the final one in $\ttc_1$.

\item
using a macro-step that simulates $q_1\step{\hopushn}q_2$,
the last write to $\ttc_n$ is a safe $u\in E_n$, and the last
write to $\ttc_{n+1}$ is $\dol_{n+1}$.

\item
using a macro-step that simulates $q_3\step{\hopopn}q_4$,
the last write to $\ttc_n$ is a safe $u\in E_n$, which is also
the last read from $\ttc_{n+1}$, implying that the final $\dol_{n+1}$
in $\ttc_{n+1}$ cannot have been read.\qedhere
\end{enumerate}
\end{proof}

\begin{lemma}
\label{lem-lifting-macro-step}
If $D$ is safe and $D\wstep{+}\wCp$ is a write-superseding
macro-step in $\wS$, then there is a $C\in\Conf_S$ such that
$D\hgeq \wC$ and $C\step{+}C'$ in $S$.
\end{lemma}
\begin{proof}
Write $D=(q,w_1,\ldots,w_k)$ and $\wCp=(q',v_1,\ldots,v_k)$. For each
$n=1,\ldots,k$, we know that $w_n$ is some $w'_n\dol_n$ (since $D$ is
safe) and that $v_n$ is some $\enc{y_n}_n$. We now consider four
cases for the macro-step $D\wstep{+}\wC$:
\begin{enumerate}[label=$\delta_{\arabic*}$:]
\item
it uses some $q=q_1\step{\ttc_1!a_i}\step{\ttc_1!\dol_1}q_2=q'$. From the
definition of $\wstep{}$ (see \autoref{ssec-write-superseding}) we
deduce that $v_1=x\cdot\dol_1$ where $x$ is a prefix of
$w'_1\cdot \dol_1\cdot a_i$ while $w_n=v_n$ for $n>1$. Since $v_1\in
E_1$, either $v_1=w_1\cdot a_i\cdot \dol_1$ and we take $\wC=D$, or
$v_1$ is a safe prefix of $w_1$, in which case we take
$C=(q,y_1,\ldots,y_n)$.

\item
it uses some $q=q_3\step{\ttc_1?\dol_1}\step{\ttc_1?a_j}q_4=q'$ with no
writing. The write-superseding semantics entails $w_1=\dol_1\cdot
a_j\cdot v_1$. Here $w_1=\enc{a_j\cdot y_1}_1$ and we take $\wC=D$.

\item
the macro-step writes a $\dol_{n+1}$ to $\ttc_n$ and reads it back, so
that $u=w_n$ and we deduce that $w_n\in E_n$ and is some $\enc{x}_n$.
Furthermore $v_n$ is $u$ perhaps after some superseding and one
obtains $x\geq_{*n}y_n$ from $u\hgeq v_n$.

On $\ttc_{n+1}$ the macro-step writes $u\cdot\dol_{n+1}$ and we reason
as in case $\delta_1$: if write-superseding erases the $\dol_{n+1}$
that closes $w_{n+1}$ then $v_{n+1}\hleq w_{n+1}$ and we set
$C=(q,\ldots,y_{n-1},x,y_{n+1},\ldots)$, otherwise
$v_{n+1}=w_{n+1}\cdot \enc{x'}_n\dol_{n+1}$ with $\enc{x'}_n\hleq u$,
 we know that $w_{n+1}$ is some $\enc{x_{n+1}}_{n+1}$ and we may set
$C=(q,\ldots,y_{n-1},x,x_{n+1},\ldots)$.

\item
On $\ttc_{n+1}$ the macro-step just reads $\dol_{n+1}\cdot u$, where
$u\in E_n$ is some $\enc{x}_n$. Necessarily $w_n$ is $\enc{x\cdot
y_{n+1}}_{n+1}$. On $\ttc_n$, one writes $\dol_{n+1}\cdot u$ and reads
$\dol_{n+1}$. Necessarily $y_n\hleq u$, hence $y_n\holeq x$. Taking
$C=(q,\ldots,y_{n-1},\epsilon,x\cdot y_{n+1},y_{n+2},\ldots)$ works.
\qedhere
\end{enumerate}
\end{proof}

\noindent We can now conclude our correctness proof.
\begin{proof}[Proof of \autoref{prop-holcs-simul-pcs}]
Assume $\wC\step{*}\wCp$. We first apply \autoref{lem-hh-to-ww} and
deduce the existence of a write-superseding run $D_0\wstep{+}\wCp$ for
some $D_0\hleq \wC$. Let us single out the stable configurations along
this run and write it under the form $D_0\wstep{+}D_1\wstep{+}\cdots
D_{r-1}\wstep{+}D_r=\wCp$, \textit{i.e.}, as a sequence of $r$ macro-steps.

We now reason by induction on $r$. If $r=0$, then $D_0=\wCp$ and
$\wC\hgeq\wCp$, implying $C\hogeq C'$ by \autoref{lem-encoding-is-iso}
so that $S$ has $C\step{*}C'$ via lossy steps.

If $r>0$, we first observe that $D_0$ is safe (since $\wC$ is) hence
$D_1, \ldots, D_r$ too by \autoref{lem-safe}. Using
\autoref{lem-lifting-macro-step} on $D_{r-1}\wstep{+}\wCp$ implies
that there is some $C''\in\Conf_S$ with $D_{r-1}\hgeq \wCpp$ and
$C''\step{+}C'$. On the other hand, the run
$\wC\step{*}D_0\wstep{+}D_{r-1}\step{*} \wCpp$ can be transformed into
$r-1$ macro-steps. We can thus apply the ind.\  hyp.\  and deduce that
$C\step{*}C''$ in $S$,
\end{proof}
We may now state formally the main theorem of this section.
\begin{theorem}
\label{thm-pcs-simul-holcs}
There is a \textsc{LogSpace} reduction that transforms reachability problems on
$k$th-order HOLCSs to reachability problems on PCSs of level $d=k+1$.
\end{theorem}
\begin{proof}[Proof Sketch]
The above ideas use $p+k$ priority levels. One can tighten these
simulations to use fewer priorities and complete the proof
of \autoref{thm-pcs-simul-holcs} by encoding the messages
$a_0,\dots,a_{p-1}$ as \emph{fixed length} binary strings over
$\{0,1\}$ followed by a $\dol_1$ separator.  Then the prioritized
alphabet $\{0,1,\dol_1,\dots,\dol_k\}$ with $k+2$ priority levels
suffices.

In particular $\{0,1,\dol\}$ is enough for LCSs. In the case of
\emph{weak} LCSs where the set of messages is linearly ordered (say
$a_0<a_1<\cdots<a_{p-1}$) and where, in addition to message losses,
any message can be replaced by a lower message inside the channels, we
can further tighten this to $\{0,\dol\}$ with a unary encoding of
message $a_i$ as $0^i\dol$.
\end{proof}
\begin{remark}
The simulation of HOLCSs by PCSs is quite straightforward. We believe
that a reduction from PCS reachability to HOLCS reachability must
exist (on complexity-theoretical grounds) but we do not have at the
moment any suggestion for a simple encoding of PCSs in HOLCSs.
\qed
\end{remark}

}%

\section{Applications of the Priority Embedding to Trees}\label{ssec-trees}
In this section we show how tree orderings can be reflected into
sequences over a priority alphabet. This %
illustrates the ``power'' of priority embeddings, and yields as a byproduct a
proof that strong tree embeddings form a wqo.  %
The consequence will be that PCSs can perform operations on encodings
of trees in a \emph{robust} way, \emph{i.e.}, such that superseding
steps respect the tree ordering.  This was already the key insight
\begin{itemize}
\item in \autoref{sec-hardy} when we encoded ordinals, seen as terms,
  into sequences over $\Sigma_d^\ast$, and
\item in \autoref{ssec-holcs} when we encoded nested sequences over
  $\Sigma^{n\ast}$, which can be seen as terms of bounded depth, also
  into sequences over $\Sigma_d^\ast$.
\end{itemize}

\noindent We generalize these ideas here, and allow in particular for labeled
trees.  More precisely, we show
\begin{enumerate}
\item in \autoref{sub-btree} that priority embeddings reflect the
  so-called \emph{strong tree-embedding} ordering on trees of bounded
  depth, and
\item in \autoref{app-tree-minors} that they reflect
  the \emph{immersion ordering} used by Gupta~\cite{gupta92} to
  compute the maximal order type of the \emph{tree minor ordering}.
\end{enumerate}
These two tree orderings are already known to be wqos---with a maximal
order type $\varepsilon_0$ that matches that of priority embeddings.
The point of the section is thus rather to show how to encode trees
robustly in priority strings, and thus how to manipulate trees
robustly in PCSs.

\begin{remark}[Kruskal's Tree Theorem]
  The reader is likely to be already acquainted with the homeomorphic
  embedding ordering used in Kruskal's Tree Theorem.  The strong tree
  embedding refines this ordering but is restricted to trees of
  bounded depth (it is not a wqo on general trees), while the tree
  minor ordering is coarser than the homeomorphic embedding ordering.
  Here, we wish to explicitly emphasize that reflections in priority embeddings
  do \emph{not} provide a new proof of Kruskal's Tree Theorem.  In
  fact, such a proof cannot exist since the maximal order type for
  homeomorphic embeddings is considerably larger than
  $\varepsilon_0$~\cite{rathjen93}.
\end{remark}

\subsection{Reflecting Bounded Depth Trees}\label{sub-btree}
Recall  \autoref{fc-refl} about order reflections: if $B$ reflects $A$ 
and $(B,{\leq_B})$ is a wqo, then
$(A,{\leq_A})$ is necessarily a wqo.  Our goal is to show that
$(\Sigma_{d,\Gamma},\gpleq)$ reflects $\Gamma$-labeled trees of depth at most $d+1$
endowed with the \emph{strong tree-embedding} relation.

Given an alphabet $\Gamma$, we note $T(\Gamma)$ for the set of finite,
ordered, unranked labeled trees (aka variadic terms) over $\Gamma$. 
Let $d$ be a depth in $\+N$.  We work here with the set $T_d(\Gamma)$
of \emph{trees of depth at most $d$} over $\Gamma$.  Formally,
$T_d(\Gamma)$ is defined by induction over $d$ by
$T_0(\Gamma)\eqdef\emptyset$, and for $d>0$, $T_d(\Gamma)$ is the
smallest set containing $T_{d-1}(\Gamma)$ and such that, if
$t_1,\dots,t_n$ are trees in $T_{d-1}(\Gamma)$ and $f$ is in $\Gamma$,
then the tree $f(t_1\cdots t_n)$ obtained by adding an $f$-labeled
root over them is in $T_d(\Gamma)$.  When $n=0$ we write $f$
rather than $f()$.  \autoref{fig-labeled-trees} presents two labeled
trees of depth $3$.
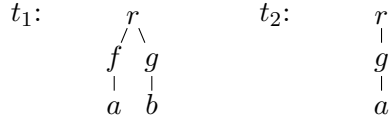
\begin{figure}[htb]
\centering
\begin{tikzpicture}[inner sep=1.5pt,level distance=.6cm,sibling
distance=.5cm,every node/.style={text height=1.5ex}]
  \node(z){$r$}
    child{node{$f$}
      child{node{$a$}}}
    child{node{$g$}
      child{node{$b$}}};
  \node[right=3cm of z](a){$r$}
    child{node{$g$}
      child{node{$a$}}};
  \node[left=of z]{$t_1$:};
  \node[left=of a]{$t_2$:};
\end{tikzpicture}
  \caption{\label{fig-labeled-trees}Two trees in $T_3(\{a,b,f,g,r\})$.}
\end{figure}

It will be convenient in the following to use the \emph{extension}
operation ``$@$'' on trees, which is defined for $n\geq 0$ by
\begin{align}
  f(t_1\cdots t_n)\mathbin{@}t&\eqdef f(t_1\cdots
  t_nt)\:;
\end{align}
in particular, $f\mathbin{@}t=f(t)$.  For instance, $t_1$
in \autoref{fig-labeled-trees} can be decomposed as
$r\mathbin{@}(f\mathbin{@}a)\mathbin{@}(g\mathbin{@}b)$.

In case where $\Gamma$ is a singleton, we denote by ``$\bullet$'' its only
element and write $T_d$ for $T_d(\{\bullet\})$.  For
instance, %
$T_1=\{\bullet\}$ contains a single tree.

\subsubsection{Strong Tree Embeddings}
Assume that $(\Gamma,\leq_\Gamma)$ is a wqo, and that we have already
defined a well-quasi-ordering $\embed_T$ on trees of maximal depth
$d$---note that as a base case, since $T_0(\Gamma)$ is empty, it is
vacuously well-quasi-ordered by the empty relation.  We can lift it
into a wqo $(T_{d+1}(\Gamma),\embed_T)$ on trees of maximal depth
$d+1$ by $f(t_1\cdots t_n)\embed_T f'(t'_1\cdots t'_m)$ $\equivdef$
$f\leq_\Gamma f'$ and $t_1\cdots t_n\embed_{T\ast}t'_1\cdots t'_m$,
\emph{i.e.}, by considering a product between $(\Gamma,\leq_\Gamma)$ and the
sequence embedding ordering on tree sequences
$(T_d(\Gamma)^\ast,\embed_{T\ast})$: by Dickson's Lemma and Higman's
Lemma, this defines a wqo on trees for every finite $d$, which we call
the \emph{strong tree embedding}.  Put differently $t\embed_T t'$ if
$t$ it can be obtained from $t'$ by deleting whole subtrees and/or
decrementing node labels.

Strong tree embeddings refine the \emph{homeomorphic tree embeddings}
used in Kruskal's Tree Theorem; in general they do not give rise to a
wqo, but in the case of bounded depth trees they do.  The two trees
$t_1$ and $t_2$ in \autoref{fig-labeled-trees} are \emph{not} related
by any homeomorphic tree embedding, and thus neither by strong tree
embedding.  The tree $t_2$ homeomorphically embeds into $t_3=r(g(g(a)))$,
but does not strongly embed into $t_3$.

Observe that $\embed_T$ is a precongruence for $@$:
\begin{align}
  t_1\embed_T t'_1\text{ and }t_2\embed_T t'_2&\text{ imply
  }t_1\mathbin{@}t_2\embed_T t'_1\mathbin{@}t'_2\:,\label{eq-ext-se}\\
  t&\embed_T t\mathbin{@}t'\:.\label{eq-ext-em}
\end{align}

\subsubsection{Encoding Trees as Strings}
It is easy to encode trees into finite sequences.  For instance,
drawing inspiration from the ordinal encodings employed
in \autoref{sec-hardy}, one might be tempted to encode the two trees
in \autoref{fig-labeled-trees} by
\begin{align}
  s_2(t_1)&=(0,a) (1,f) (0,b) (1,g) (2,r)\\
  s_2(t_2)&=(0,a) (1,g) (2,r)
\end{align}
as the result of an inductive encoding $s_d(w(t_1\cdots t_n))\eqdef
s_{d-1}(t_1)\cdots s_{d-1}(t_n)\cdot(d,w)$.  Observe however that we
wish our encoding to be an order reflection
(\textit{cf.}\ \autoref{ssec-refl}), which is not the case with $s_d$:
we see that $s_2(t_2)\gpleq s_2(t_1)$, although $t_2\not\embed_T
t_1$.  Over a singleton alphabet however, $s_d$ \emph{is} an order
reflection from $T_{d+1}$ to $\Sigma_d^\ast$.

Here we present a more redundant encoding, apt to handle arbitrary
alphabets.  We encode trees of bounded depth using the mapping
$\enc{.}_d{:}\,T_{d+1}(\Gamma)\to\Sigma_{d,\Gamma}^\ast$ defined by
induction on $d$ by
\begin{align}
 \enc{w(t_1\cdots t_n)}_d&\eqdef\begin{cases}
 (d,w)&\text{if }n=0\;,\\
  \enc{t_1}_{d-1}\cdot (d,w)\cdots
  \enc{t_n}_{d-1}\cdot(d,w)&\text{otherwise.}
 \end{cases}
\end{align}
For instance, if we fix $d=2$, the trees
in \autoref{fig-labeled-trees} are encoded as
\begin{align}
\enc{t_1}_2&= (0,a)(1,f)(2,r)(0,b)(1,g)(2,r)\;,\\
\enc{t_2}_2&= (0,a)(1,g)(2,r)\;.
\end{align}
This satisfies $\enc{t_2}_2\not\gpleq \enc{t_1}_2$ as desired.

\subsubsection{Proper Words}
Not every string in $\Sigma^\ast_{d,\Gamma}$ is the encoding of a tree
according to $\enc{.}_d$: for $0\leq a \leq d$, we let
\begin{align}
  P_{-1,\Gamma}&\eqdef\emptyset\;,
  &P_{a,\Gamma}&\eqdef \bigcup_{w\in\Gamma}\big(P_{a-1,\Gamma}\cdot\{(a,w)\}\big)^\ast\cdot(P_{a-1,\Gamma}\cup\{\epsilon\})\cdot\{(a,w)\}
\end{align}
be the set of \emph{proper encodings of height $a$}.  Then
$P\eqdef\bigcup_{a\leq d}P_{a,\Gamma}$ is the set of \emph{proper}
words in $\Sigma^\ast_{d,\Gamma}$.
A proper word $x$ belongs to a unique $P_{a,\Gamma}$ with $a=h(x)$, where
$h(x)$ is the height of $x$,
and has then a canonical factorization of the form $x=x_1(a,w)\cdots
x_m(a,w)$ with every $x_j$ in $P_{a-1,\Gamma}$ and $w$ in
$\Gamma$.

Given a depth $d$, we see that $\enc{.}_d$ is a bijection between
$T_{d+1}(\Gamma)$ and $P_{d,\Gamma}$, with inverse
$\tau{:}\,P_{d,\Gamma}\to T_{d+1}(\Gamma)^\ast$ defined using
canonical decompositions by
\begin{align}
  \tau(x=x_1(h(x),w)\cdots x_m(h(x),w))&\eqdef w(\tau(x_1)\cdots\tau(x_m))
  \\&= w\mathbin{@}\tau(x_1)\mathbin{@}\cdots\mathbin{@}\tau(x_m)\;.\label{eq-tau-ext}
\end{align}

\begin{proposition}\label{prop-refl-Td-Sd}
  The map $\enc{.}_d$ is an order reflection from
  $(T_{d+1}^\ast,{\embed_{T\ast}})$ to $(\Sigma_d^\ast,{\pleq})$.
\end{proposition}

\proof
Let $x$ and $x'$ be two proper words in $P_{d,\Gamma}$ with $x\gpleq
x'$; we show by induction on $x$ that $\tau(x)\embed_T\tau(x')$.  We
consider the canonical factorizations $x=x_1(d,w)\cdots x_m(d,w)$ and
$x'=x'_1(d,w')\cdots x'_n(d,w')$ for $m,n\geq 0$, $x_j,x'_j$ in
$P_{d-1,\Gamma}$ for all $j$, and $w,w'$ in $\Gamma$.

By definition of the generalized priority embedding, the $m$ pairs
$(d,w)$ occurring in $x$ must be mapped to some pairs $(d,w')$
occurring in $x'$ with $w\leq_\Gamma w'$.  Hence there exist $1\leq
i_1,\dots,i_m\leq n$ such that $x_j\gpleq x'_{i_j}$.  Therefore,
\begin{align*}
  \tau(x)&=w\mathbin{@}\tau(x_1)\mathbin{@}\cdots\mathbin{@}\tau(x_m)
      &&\text{by \eqref{eq-tau-ext}}\\
  &\embed_T w'\mathbin{@}\tau(x_1)\mathbin{@}\cdots\mathbin{@}\tau(x_m)
      &&\text{by \eqref{eq-ext-se} since $w\leq_\Gamma w'$}\\
  &\embed_Tw'\mathbin{@}\tau(x'_{i_1})\mathbin{@}\cdots\mathbin{@}\tau(x'_{i_m})
      &&\text{by \eqref{eq-ext-se} and ind.\ hyp.\ on $x_j\gpleq x'_{i_j}$}\\
  &\embed_Tw'\mathbin{@}\tau(x'_1)\mathbin{@}\cdots\mathbin{@}\tau(x'_n)
      &&\text{by \eqref{eq-ext-em}}\\
  &=\tau(x')&&\text{by \eqref{eq-tau-ext}}\;.
  \rlap{\hbox to 157 pt{\hfill\qEd}}
\end{align*}

\noindent
Note that \autoref{prop-refl-Td-Sd} provides an alternative proof of
the fact that $(T_d(\Gamma),\embed_T)$ is a wqo, thanks
to \autoref{thm-gpe-wqo} and \autoref{fc-refl}.

\subsection{Relationship to Tree Minors}\label{app-tree-minors}
As a further application of the priority embedding, we demonstrate
that we also subsume another wqo on trees, the \emph{tree minor}
ordering, by using the techniques of \citet{gupta92} to encode trees
into generalized prioritized alphabets. The tree minor ordering is
coarser than the homeomorphic embedding, but the upside is that trees
of \emph{unbounded} depth can be encoded into strings.
\newcommand{\gordering}{\sqsubseteq_{{I}}}
\newcommand{\glinearisation}{\mathit{lin}}
\newcommand{\gdomain}{\mathcal{I}}

The trees considered in~\citep{gupta92} are unlabeled finite rooted
trees with an ordering on the children of every internal vertex,
called \emph{planar planted trees} therein, \emph{i.e.} trees from
$T\eqdef T(\{\bullet \})$. \autoref{fig-two-trees} illustrates two
such trees. The ordering on the children in particular implies that,
for instance, the tree $\bullet(\bullet(\bullet), \bullet)$ is not
equivalent to the tree $\bullet(\bullet, \bullet(\bullet))$. Gupta
gives in~\citep{gupta92} a constructive proof that planar planted
trees are well-quasi-ordered under minors. Recall that $t_1$ is a
\emph{minor} of $t_2$ if $t_1$ can be obtained from $t_2$ by a series
of edge contractions, \emph{e.g.}\ in \autoref{fig-two-trees}, the
left-hand tree is a minor of the right-hand one. Note that, however, the two
trees are incomparable for the previously considered homeomorphic
embeddings.
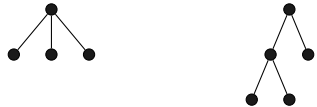
\begin{figure}[htbp]
  \centering
  \begin{tikzpicture}[inner sep=1.5pt,level distance=.6cm,sibling distance=.5cm,every
      node/.style={draw,circle,fill=black!90}]%
    \node(a){}
    child{node{}}
    child{node{}}
    child{node{}};
    \node[right=3cm of a]{}
    child{node{}
      child{node{}}
      child{node{}}}
    child{node{}};
  \end{tikzpicture}
  \caption{\label{fig-two-trees}Two trees in $T_2$.}
\end{figure}

Gupta provides in \citep{gupta92} an effective linearization
$\glinearisation{:}\,T \to \bigcup_{d\ge 0} \Sigma_{d,\Gamma}^*$ which
associates with every tree $t$ a word $\glinearisation(t)$ over the
generalized prioritized alphabet $\Sigma_{d,\Gamma}$, where $d$ is
dubbed the \emph{width} of $t$---which is at most its number of
vertices---and $(\Gamma,=)$ is a finite alphabet with $\Gamma=\{0, 1,
2, 3\}$. Gupta continues by defining a so-called \emph{immersion
  ordering $\gordering$} on $\Sigma_{d,\Gamma}^*$ for any fixed $d\in
\mathbb{N}$ as follows: given $x=(a_1,w_1)(a_2,w_2)\cdots (a_k,
w_k)\in \Sigma_{d,\Gamma}^*$ and $y\in \Sigma_{d,\Gamma}^*$,
$x\gordering y$ if $y$ can be factored as $y=y_0 y_1\cdots y_k
y_{k+1}$ such that
\begin{align*}
  \label{eqn:immersion-ordering}
  y_i\in (\Sigma_{d,\Gamma}\setminus \Sigma_{d-a_i-1,\Gamma})^*\cdot
  (a_i,w_i) \cdot (\Sigma_{d,\Gamma}\setminus
  \Sigma_{d-a_i-1,\Gamma})^*,~ 1\le i\le k.
\end{align*}
The crucial relationship between trees in $T$, $\glinearisation$ and
$\gordering$ is established
in~\citep[\theoremautorefname~4.1]{gupta92}: given planar planted
trees $t_1,t_2\in T$, whenever $\glinearisation(t_1)\gordering
\glinearisation(t_2)$ then $t_1$ is a minor of $t_2$. By showing that
$\gordering$ is a well-quasi-ordering, Gupta concludes that planar
planted trees are well-quasi-ordered under minors.

The immersion ordering $\gordering$ is closely related to our
generalized priority ordering. In fact, it is easily seen that
$\gpleq$ can be viewed as a sub-structure of $\gordering$. Define
an automorphism $\kappa:\Sigma_{d,\Gamma} \to \Sigma_{d,\Gamma}$ as
\begin{align*}
  \kappa(a,w) & \eqdef (d-a,w)
\end{align*}
which canonically extends to words over $\Sigma_{d,\Gamma}^*$. Now
$x\gpleq y$ in particular implies $\kappa(x) \gordering
\kappa(y)$. Thus, our \autoref{thm-gpe-wqo} yields as a corollary another
proof that the immersion ordering $\gordering$ is a wqo.
\begin{corollary}
  The immersion ordering $\gordering$ is a well-quasi-ordering.
\end{corollary}

\subsection{Further Applications}
As stated in the introduction to this section, our main interest in
strong tree embeddings is in connection with structural orderings of
ordinals; see \autoref{sec-hardy}.
Bounded depth trees are also used
in the verification of infinite-state systems as a means to obtain
decidability results, in particular for tree pattern rewriting
systems~\citep{genest08} in XML processing, and, using elimination
trees~\citep{sparsity}, for bounded-depth graphs used \textit{e.g.}\ in
the verification of \textit{ad-hoc} networks~\citep{delzanno10}, the
$\pi$-calculus~\citep{meyer08}, programs~\citep{bansal13}, and
protocols~\citep{konig14}.
These applications consider \emph{labeled} trees, which motivate the
generalized priority alphabets and embedding defined
in \autoref{sec-embedding}.

The exact complexity of verification problems in the aforementioned
models is currently
unknown~\citep{genest08,delzanno10,meyer08,bansal13,konig14}.  Our
encoding suggests they might be $\F_{\ez}$-complete.  We hope to see
PCS Reachability employed as a ``master'' problem for $\F_{\ez}$%
, like LCS Reachability for $\F_{\omom}$, which is
used in reductions instead of more difficult proofs based on Turing
machines and Hardy computations.

\section{Concluding Remarks}
\label{sec-concl}
We introduced Priority Channel Systems, a natural model for protocols
and programs with differentiated, prioritized asynchronous
communications, and showed how they give rise to well-structured
systems with decidable model-checking problems.

We showed that Reachability and Termination for PCSs are
$\F_{\ez}$-complete, and we expect our techniques to be transferable
to other models, \textit{e.g.}\ models based on wqos on bounded-depth trees or
graphs, whose complexity has not been
analyzed~\citep{genest08,delzanno10,meyer08,bansal13,konig14}. This is part of
our current research agenda on complexity for well-structured systems~\cite{SS-esslli2012}.

In spite of their enormous worst-case complexity, we expect PCSs to be
amenable to regular model checking techniques \textit{\`a
la}~\cite{abdulla96b,boigelot99b}. This requires investigating the
algorithmics of upward- and downward-closed sets of configurations
wrt.\  the priority ordering. These sets, which are always regular,
seem promising since $\pleq$ shares some good properties with the
better-known subword ordering, \textit{e.g.}\ the upward- or downward-closure
of a sequence $x\in\Sigma^\ast_d$ can be represented by a
deterministic finite automaton with $\len{x}$ states.

\section*{Acknowledgments}
We thank Lev Beklemishev who drew our attention to~\cite{schutte85}
and the reviewers for their helpful suggestions.

\bibliographystyle{alpha}
\bibliography{journals,conferences,plcm}
\vspace{-20 pt}
\end{document}